\documentclass[11pt]{article}
\usepackage{amsopn}
\usepackage{amssymb, amscd}
\usepackage{mathrsfs}

\usepackage{amsmath,amssymb}

\usepackage[english]{babel}

\usepackage{graphics,color}

\usepackage{amsthm}
\usepackage{amsmath}

\usepackage{graphicx}

\newcommand{\nc}{\newcommand}

\nc{\SO}{\mathrm{SO}} \nc{\Spe}{\mathrm{Sp}} \nc{\Sl}{\mathrm{SL}}
\nc{\SU}{\mathrm{SU}} \nc{\Or}{\mathrm{O}} \nc{\U}{\mathrm{U}}
\nc{\Gl}{\mathrm{GL}} \nc{\Se}{\mathrm{S}} \nc{\Cl}{\mathrm{Cl}}
\nc{\Spein}{\mathrm{Spin}} \nc{\Pin}{\mathrm{Pin}}
\nc{\G}{\mathrm{GL}_n(\RR)} \nc{\g}{\mathfrak{gl}_n(\RR)}

\nc{\vs}{\vspace{.2cm}} \nc{\vsp}{\vspace{1cm}}

\nc{\tq}{T_q}

\nc{\sqa}{\mathcal{S}^{a}_q}
\nc{\sq}{\mathcal{S}_q}
\nc{\sqao}{\mathcal{S}^{a\mathscr{O}}_q}
\nc{\sqo}{\mathcal{S}^{\mathscr{O}}_q}
\nc{\sqan}{\mathcal{S}^{a}_{q,N}}
\nc{\sqn}{\mathcal{S}_{q,N}}
\nc{\sqe}{\sq^{\epsilon,r}}
\nc{\sqh}{\widehat{\sq^{\epsilon,r}}}
\nc{\sqns}{\mathcal{S}^{\sigma,N}_{q,N}}
\nc{\sqnh}{\widehat{\mathcal{S}^{\sigma,N}_{q,N}}}
\nc{\sqno}{\sqn^{\O}}
\nc{\sqnso}{(\mathcal{S}^{\sigma,N}_{q,N})^{\O}}
\nc{\sqnoa}{\sqn^{\O^{a}}}
\nc{\sqnoh}{\widehat{\sqno}}
\nc{\sqnsn}{\mathcal{S}^{\sigma,n}_{q,N}}
\nc{\sqnsnn}{\mathcal{S}^{\sigma,N}_{q,N}}

\nc{\mtir}{\widetilde{M}(\infty, R)}
\nc{\mtim}{\widetilde{M}(\infty)[m]}
\nc{\ri}{R^{\infty}}
\nc{\zz}{[z,z^{-1}]}
\nc{\vsm}{\varphi^{[m]}_s}

\nc{\enz}{\in \mathbb{Z}}
\nc{\glh}{\widehat{ g\ell}_\infty^{[m]}}
\nc{\gl}{{ g\ell}_\infty^{[m]}}
\nc{\bi}{b^{[m]}_{\infty}}
\nc{\ci}{c^{[m]}_{\infty}}
\nc{\di}{d^{[m]}_{\infty}}

\nc{\fims}{\varphi^{[m]}_s}
\nc{\gm}{\mathfrak{g}^{[m]}}

\numberwithin{equation}{section}

\theoremstyle{plain}
\newtheorem{theorem}{Theorem}[section]
\newtheorem{proposition}[theorem]{Proposition}
\newtheorem{corollary}[theorem]{Corollary}
\newtheorem{lemma}[theorem]{Lemma}

\theoremstyle{definition}
\newtheorem{definition}[theorem]{Definition}

\theoremstyle{remark}
\newtheorem{remark}[theorem]{Remark}

\numberwithin{theorem}{section}

\newcommand{\CC}{{\mathbb C}}
\newcommand{\RR}{{\mathbb R}}

\newcommand{\ZZ}{{\mathbb Z}}
\newcommand{\NN}{{\mathbb N}}

\renewcommand{\O}{\mathcal{O}}

\begin{document}
%

\title{QHWM of the orthogonal and symplectic types Lie subalgebras of the Lie algebra of the matrix quantum pseudo differential operators}

\author{Karina Batistelli and Carina Boyallian\thanks{%
     Ciem - FAMAF, Universidad Nacional de C\'ordoba - (5000) C\'ordoba,
Argentina
\newline $<$batistelli@famaf.unc.edu.ar - boyallia@mate.uncor.edu$>$.}}
\date{}

\maketitle

\begin{abstract}
In this paper we classify the irreducible quasifinite highest weight modules over the orthogonal and syplectic types Lie subalgebras of the Lie algebra of the matrix quantum pseudo differential operators. We also realize them in terms of the irreducible quasifinite highest weight modules of the Lie algebras of infinite matrices with finitely many nonzero diagonals and its classical Lie subalgebras of types B, C and D.
\end{abstract}

\vfill \pagebreak

\section{Introduction}
The study of $W$-infinity algebras has its origins in various physical theories, such as conformal field theory, the theory of quantum Hall effect, etc.  The most important of these algebras is $W_{1+\infty}$, which is the central extension of the Lie algebra $D$ of differential operators on the circle.

The dificulty when studying the representation theory of these algebras lies in the fact that, although they admit a $\ZZ$-gradation and a triangular decomposition, each of the graded subspaces is still infinite dimensional. As a consequence, the study of highest weight modules that satisfy the quasifinite condition, which is, graded subspaces are finite dimensional, becomes a nontrivial problem.

The representations of the Lie algebra $W_{1+\infty}$ were first studied in \cite{KR1}, where its irreducible quasifinite highest weight modules were characterized and it was shown that they can be realized in terms of the irrducible highest weight representations of the Lie algebra of infinite matrices. At the end of that article, similar results were found for the central extension of the Lie algebra of quantum pseudo-differential operatos $\sq$, which contains as a subalgebra the $q$-analogue of the Lie algebra $\widehat{D}$, the algebra of all regular difference operators on $\mathbb{C}^{\times}$.

This study for $\widehat{D}$ was continued in \cite{FKRW}, \cite{KL} and \cite{KR2} in the framework of vertex algebra theory and in \cite{BKLY} for the matrix case. In \cite{KL}, V. Kac and J. Liberati also gave some general results on the characterization of quasifinite representations of any $\ZZ$-graded Lie algebra, which will be used in this paper.
In \cite{KWY}, a classification was given of the irreducible quasifinite highest weight modules of the central extension of the Lie subalgebras of $D$ fixed by minus the anti-involutions preserving the principal gradation. These results were extended in \cite{BL1} to the algebra $D^N$ of the $N \times N$-matrix differential operators on the circle.

An analogous study was carried out for the Lie algebra of quantum pseudo-differential operators. In \cite{BL3} it was shown that there is a family of anti-involutions on $\sq$, up to conjugation, preserving the principal gradation. Their irreducible quasifinite highest weight modules were classified and realized in terms of irreducible highest weight representations of the Lie algebra of infinite matrices with finitely many nonzero diagonals $\gl$ and its classical Lie subalgebras of B, C and D types. Similarly, in \cite{BL2},   the quasifinite highest weight modules over the central extension of the Lie algebra of $N \times N$ matrix quantum pseudo differential operators, denoted $\widehat{\sqn}$, were classified and characterized them in terms of the representation theory of the Lie algebra of infinite matrices with finitely many nonzero diagonals.

Making use of the the description of Lie subalgebras of $\widehat{\sqn}$ fixed by minus the anti-involutions preserving the principal gradation given in \cite{BB}, we classify the irreducible highest weight modules of some of the subalgebras found, particularly the orthogonal and symplectic types. This paper is organized as follows. In Sect. 2 we present some standard facts of representation theory of $\glh$ and its subalgebras of types B, C and D. In Sects. 3 and 4 we introduce the subalgebras $\sqnsnn$ and we study the structure of its parabolic subalgebras. In Sect. 5 we give a characterization of the irreducible quasifinite highest weight modules of $\sqnh$. In Sect. 6 an interplay between $\sqnh$ and the infinite rank classical Lie algebras of types A, B, C and D is established. Finally, in Sect. 7 we give the realization of the irreducible quasifinite highest weight modules of $\sqnh$.

\section{Lie algebras $\glh$ and its classical Lie subalgebras}

In this section we will give a description of the  Lie algebra of infinite matrices with finitely many nonzero diagonals $\gl$ and its classical Lie subalgebras of B, C and D types. We will follow the notation in Sect. 1 of \cite{KWY}.

Denote $R_m=\CC[u]/(u^{m+1})$ the quotient algebra of the polynomial algebra $\CC[u]$ by the ideal generated by $u^{m+1}$ ($m \in \ZZ_{\geq 0}$). Let $\mathbf{1}$ be the identity element in $R_m$. Denote by $\gl$ the complex Lie algebra of all infinite matrices $(a_{i,j})_{i,j \in \ZZ}$ with only finitely many nonzero diagonals with entries in $R_m$. Denote $E_{i,j}$ the infinite matrix with $1$ at $(i,j)$-entry and $0$ elsewhere. There is a natural automorphism $\nu$ of $\gl$ given by

\begin{equation}\label{eq:nu}
\nu(E_{i,j})=E_{i+1,j+1}.
\end{equation}

Let the weight of $E_{i,j}$ be $j-i$. This defines the \textit{principal} $\ZZ$-gradation $\gl=\oplus_{j \in \ZZ}(\gl)_j$. Denote by $\glh=\gl \oplus R_m$ the central extension of $\gl$ given by the following $2$-cocycle with values in $R_m$:

\begin{equation}{\label{eq:2cocycleb}}
C(A,B)= Tr([J,A]B),
\end{equation}
where $J=\sum_{i \leq 0} E_{i,i}$. The $\ZZ$-gradation of the Lie algebra $\gl$ extends to $\glh$ by putting the weight of $R_m$ to be $0$. In particular, we have the \textit{triangular decomposition},

\begin{equation}
\glh = (\glh)_- \oplus (\glh)_0 \oplus (\glh)_+,
\end{equation}
where

\begin{equation}
(\glh)_\pm = \oplus_{j \in \NN} (\glh)_{\pm j} \quad \text{and} \quad (\glh)_0=(\gl)_0 \oplus R_m.
\end{equation}

Given $\lambda \in (\glh)^*_0$, we let
\begin{eqnarray}\label{eq:gl}
 c_i &=& \lambda(u^i), \\ \nonumber
 \ ^a\lambda^{(i)}_j  &=&  \lambda (u^i E_{j,j}), \\ \nonumber
 \ ^aH^{(i)}_j &=& u^iE_{j,j}-u^iE_{j+1,j+1} + \delta_{j,0}c_i, \\ \nonumber
 \ ^ah^{(i)}_j &=&  \ ^a\lambda^{(i)}_j - \ ^a\lambda^{(i)}_{j+1} + \delta_{j,0}c_i, \nonumber
\end{eqnarray}
where $j \in \ZZ$ and $0 \leq i \leq m.$ Let $L(\glh,\lambda)$ be the irreducible highest weight $\glh$-module with highest weight $\lambda$. The $^a\lambda^{(i)}_j$ are called \textit{labels} and $c_i$ are the \textit{central charges} of $L(\glh, \lambda)$.


Consider the vector space $R_m[t, t^{-1}]$ and take the basis $v_i=t^{-i}$, $i \in \ZZ$ over $R_m$. Now consider the following $\CC$-bilinear form on $R_m[t, t^{-1}]$:

\begin{equation}
B^{\pm}(u^mv_i,u^nv_j)=u^m(-u)^{-n}(\pm1)^{i}\delta_{i,-j}.
\end{equation}
Denote by $\bar{b}^{- [m]}_{\infty}$ (resp. $\bar{b}^{+ [m]}_{\infty}$) the Lie subalgebra of $\gl$ which preserves the bilinear form $B^-(\cdot,\cdot)$ (resp. $B^+(\cdot,\cdot)$). We have
\begin{align*}
& \bar{b}^{+ [m]}_{\infty}=\{ (a_{i,j}(u))_{i,j \in \ZZ} \in \gl : a_{i,j}(u)=-a_{-j,-i}(-u) \}, \\
& \bar{b}^{- [m]}_{\infty}=\{ (a_{i,j}(u))_{i,j \in \ZZ} \in \gl : a_{i,j}(u)=(-1)^{1+i+j}a_{-j,-i}(-u) \}.
\end{align*}
Denote by $\bi=\bar{b}^{-[m]}_{\infty} \oplus R_m$ (resp. $\tilde{b}^{[m]}_{\infty}=\bar{b}^{+[m]}_{\infty} \oplus R_m$) the central extension of $\bar{b}^{-[m]}_{\infty}$ (resp. $\bar{b}^{+[m]}_{\infty}$) given by the restriction of the $2$-cocycle \eqref{eq:2cocycleb}, defined in $\gl$. The subalgebra $\bi$ (resp.  $\tilde{b}^{[m]}_{\infty}$) inherits from ${\glh}$ the principal $\ZZ$-gradation and the triangular decomposition (see \cite{KR2} and \cite{FKRW} for notation),
\begin{align*}
& \bi = \oplus_{j \in \ZZ}(\bi)_j, \quad \bi=(\bi)_+ \oplus (\bi)_0 \oplus (\bi)_-, \\
& \tilde{b}^{[m]}_{\infty} = \oplus_{j \in \ZZ} (\tilde{b}^{[m]}_{\infty})_j , \quad \tilde{b}^{[m]}_{\infty}=(\tilde{b}^{[m]}_{\infty})_+ \oplus (\tilde{b}^{[m]}_{\infty})_0 \oplus (\tilde{b}^{[m]}_{\infty})_-
\end{align*}
Note that the Lie algebra $\tilde{b}^{[m]}_{\infty}$ is isomorphic to ${b}^{[m]}_{\infty}$ via the isomorphism that sends the elements $u^kE_{i,j}-(-u)^kE_{-j,-i}$ to $u^kE_{i,j}+(-1)^{1+i+j}(-u)^kE_{-j,-i}$, $i,j \in \ZZ$, $k \in \ZZ_{+}$. Their Cartan subalgebra coincides. In particular, when $m=0$, we have the usual Lie subalgebra of $ g\ell_\infty$, denoted $b_{\infty}$ (see \cite{K}) (resp. $\tilde{b}_{\infty}$, see \cite{W}). Given $\lambda \in (\bi)^{\ast}_0$, denote $L(\bi, \lambda)$ the irreducible highest weight module over $\bi$ with highest weight $\lambda$.

For each $\lambda \in (\bi)^{\ast}_0$, we let
\begin{align} \label{eq:bi}
& c_i=\lambda(u^i), \\ \nonumber
& ^b\lambda^{(j)}_0= \lambda(2u^jE_{0,0}) \quad (j \, \text{ odd)}, \\ \nonumber
& ^b\lambda^{(j)}_i= \lambda(u^jE_{i,i}-(-u)^jE_{-i,-i}), \\ \nonumber
& ^bH^{(j)}_i= u^jE_{i,i}-u^jE_{i+1,i+1}+(-u)^jE_{-i-1,-i-1}-(-u)^jE_{-i,-i}, \\ \nonumber
& ^bH^{(j)}_0= 2(u^jE_{0,0}-u^jE_{-1,-1}-u^jE_{1,1})+u^j, \quad (j \, \text{ even)},\\ \nonumber
& ^bH^{(j)}_0= (2u^jE_{0,0}-u^jE_{-1,-1}-u^jE_{1,1})+u^j, \quad (j \, \text{ odd)},\\ \nonumber
& ^bh^{(j)}_i= \lambda (^bH^{(j)}_i) =  {^b\lambda^{(j)}_i}-  {^b\lambda^{(j)}_{i+1}},\\ \nonumber
& ^bh^{(j)}_0= \lambda (^bH^{(j)}_0) = -2 \,\, {^b\lambda^{(j)}_1}+2c_j \quad (j \, \text{ even)},\\ \nonumber
& ^bh^{(j)}_0= \lambda (^bH^{(j)}_0) =   {^b\lambda^{(j)}_0}-  {^b\lambda^{(j)}_1}+c_j \quad (j \, \text{ odd)},
\end{align}
where $i \in \NN$ and $0 \leq j \leq m$. The $^b\lambda^{(i)}_j$ are called the \textit{labels} and $c_i$ are the \textit{central charges} of $L(\bi, \lambda)$ or  $L(\tilde{b}^{[m]}_{\infty}, \lambda)$.


Now consider the following $\CC$-bilinear form on $R_m[t,t^{-1}]$:

\begin{equation}
   C(u^m v_i,u^nv_j)= u^{m}(-u^{n})(-1)^i\delta_{i,1-j}   \, .
\end{equation}

Denote by $\bar c_{\infty}^{[m]}$  the Lie subalgebra of $g\ell_\infty ^{[m]}$ which preserves the bilinear form $C (\,,\, )$. We have

\begin{displaymath}
   \bar c_{\infty}^{[m]}= \{ (a_{ij}(u))_{i,j\in\ZZ} \in g\ell_\infty  ^{[m]}\, | \, a_{ij}(u)=(-1)^{i+j+1}a_{1-j,1-i}(-u)\, \} \, .\end{displaymath}

Denote by $c_{\infty} ^{[m]}=\bar c_{\infty} ^{[m]}\oplus R_m$ the central extension of $\bar c_{\infty} ^{[m]}$   given by the restriction of the 2-cocycle  (2.2), defined in $g\ell_\infty ^{[m]}$. This subalgebra inherits from $\widehat {g\ell}_\infty^{[m]}$ the principal $\ZZ$-gradation and the
triangular decomposition, (see \cite{KWY} and \cite{K} for notation)

\begin{displaymath}
    c_{\infty}^{[m]}=\oplus_{j\in\ZZ} (c_{\infty}^{[m]})_j\qquad c_{\infty}^{[m]}=(c_{\infty}^{[m]})_+\oplus (c_{\infty}^{[m]})_0\oplus (c_{\infty}^{[m]})_- \, .
\end{displaymath}

In particular when $m=0$, we have the usual Lie subalgebra of $g\ell _\infty$, denoted by $c_\infty$.

Given $\lambda\in (c_{\infty}^{[m]})^*_0$, denote by $L(c_{\infty}^{[m]}; \lambda)$ the irreducible highest weight module over $c_{\infty}^{[m]}$  with highest weight $\lambda$. For each $\lambda\in(c_\infty^{[m]})_0^*$, we let:

\begin{eqnarray} \label{eq:ci}
  \quad \quad c_i &= & \lambda(u^i),\nonumber\\
  \quad\quad \ ^c\lambda_j^{(i)}&=&\lambda(u^i\,E_{j,j}-(-u)^i\, E_{1-j,1-j}),\nonumber\\
   \quad\quad \ ^cH^{(i)}_j  &=&   u^iE_{j,j}-u^jE_{j+1,j+1}+(-u)^iE_{-j,-j}-(-u)^iE_{1-j,1-j}, \\ \nonumber
 \quad\quad \ ^cH^{(i)}_0  &=&  (u^iE_{0,0}-u^iE_{1,1})+u^i, \quad (i \, \text{ even)} \\ \nonumber
  \quad\quad \ ^ch_j^{(i)} &=&    ^c\lambda_j^{(i)}-\ ^c\lambda_{1+j}^{(i)},\nonumber\\
\quad \quad\ ^ch_0^{(i)} &=&    ^c\lambda_1^{(i)}+c_i\quad (i\,
\hbox{ even}),
\end{eqnarray}
where $j\in\NN$ and $i=0,\cdots,m$. For later use, it is convenient to put $^ch_0^{(i)}=c_i$ ($i$ odd), $i=0,\cdots,m$.

The $ \ ^c\lambda_j^{(i)}$ are called the {\it labels} and $c_i$ are the ¨{\it central charges} of $L( \ci,\lambda)$.


\medskip
Now consider the following $\CC-$bilinear form on $R_m[t,t^{-1}]$:
\begin{equation}
D(u^mv_i,u^nv_j)=u^m(-u)^{n}\delta_{i,1-j}.
\end{equation}
Denote by $\di$ the Lie subalgebra of $\gl$ which preserves the bilinear form $D(\cdot , \cdot)$. We have
\begin{equation*}
\bar{d}^{[m]}_{\infty}=\{ (a_{i,j}(u))_{i,j \in \ZZ} \in \gl : a_{i,j}(u)=-a_{1-j,1-i}(-u) \}.
\end{equation*}
Denote by $\di=\bar{d}^{[m]}_{\infty} \oplus R_m$ the central extension of $\bar{d}^{[m]}_{\infty}$ given by the restriction of the $2$-cocylcle \eqref{eq:2cocycleb}, defined in $\gl$. This subalgebra inherits from ${\glh}$ the principal $\ZZ$-gradation and the triangular decomposition (see \cite{KR2} and \cite{FKRW} for notation),
\begin{equation*}
\di = \oplus_{j \in \ZZ}(\di)_j, \quad \di=(\di)_+ \oplus (\di)_0 \oplus (\di)_-.
\end{equation*}
Given $\lambda \in (\di)^{\ast}_0$, denote $L(\di, \lambda)$ the irreducible highest weight module over $\di$ with highest weight $\lambda$.

For each $\lambda \in (\di)^{\ast}_0$, we let
\begin{align} \label{eq:di}
& c_i=\lambda(u^i), \\ \nonumber
& ^d\lambda^{(j)}_i= \lambda(u^jE_{i,i}-(-u)^jE_{1-i,1-i}), \\ \nonumber
& ^dH^{(j)}_i= u^jE_{i,i}-u^jE_{i+1,i+1}+(-u)^jE_{-i,-i}-(-u)^jE_{1-i,1-i}, \\ \nonumber
& ^dH^{(j)}_0= ((-u)^jE_{0,0}+(-u)^jE_{-1,-1}-u^jE_{2,2}-u^jE_{1,1})+2u^j, \\ \nonumber
& ^dh^{(j)}_i= \lambda (^dH^{(j)}_i) =  {^d\lambda^{(j)}_i}-  {^d\lambda^{(j)}_{i+1}},\\ \nonumber
& ^dh^{(j)}_0= \lambda (^dH^{(j)}_0) =   -{^d\lambda^{(j)}_1}-  {^d\lambda^{(j)}_2}+2c_j ,
\end{align}
where $i \in \NN$ and $0 \leq j \leq m$. The $^d\lambda^{(i)}_j$ are called the \textit{labels} and $c_i$ are the \textit{central charges} of $L(\di, \lambda)$. In particular, when $m=0$ we have the usual $\bar{d}_{\infty}=\bar{d}^{[0]}_{\infty}$, $d_{\infty}=d^{[0]}_{\infty}$, cf. \cite{K}. In this case, we drop the superscript $[0]$.

\section{The Lie algebra $\sqn$}\label{sec:sqn}

Consider $\mathbb{C}[z,z^{-1}]$ the Laurent polynomial algebra in one variable. 
We denote $\sqa$ the associative algebra of quantum pseudo-differential operators.
Explicitly, let $T_q$ denote the  operator on $\mathbb{C}[z,z^{-1}]$ given by

\begin{equation*}
T_q f(z)=f(qz),
\end{equation*}
where $q \, \in \, \mathbb{C}^{\times}=\mathbb{C} \backslash \{0\}$. An element of $\sqa$ can be written as a linear combination of operators of the form $z^{k}f(T_q)$, where $f$ is a Laurent polynomial in $T_q$. The product in $\sqa$ is given by
\begin{equation}
(z^{m}f(T_q))(z^{k}g(T_q))=z^{m+k}f(q^{k}T_q)g(T_q).
\end{equation}

Denote $\sq$ the Lie algebra obtained from $\sqa$ by taking the usual commutator. Take  $\sq^{\prime}:=[\sq,\sq]$. It follows:

\begin{equation*}
\sq=\sq^{\prime} \oplus \mathbb{C}T^{0}_q \quad \hbox{ (direct sum of ideals).}
\end{equation*}

Let $N$ be a positive integer. As of this point, we shall denote by $Mat_N A$ the associative algebra of all $N \times N$ matrices over an algebra $A$ and $E_{ij}$ the standard basis of $Mat_N \mathbb{C}$.

Let $\sqan=\sqa \otimes Mat_N \mathbb{C}$ be the associative algebra of all quantum matrix pseudodifferential operators, namely the operators on $\mathbb{C}^N[z, z^{-1}]$ of the form
\begin{equation}
E=e_k(z)T^k_q+e_{k-1}(z)T^{k-1}_q + \cdots + e_0(z),
\hbox{ where } e_k(z) \, \in \, Mat_N \mathbb{C}[z, z^{-1}].
\end{equation}

In a more useful notation, we write the matrix of pseudodifferential operators as linear combinations of elements of the form $z^k f(T_q)A$, where $f$ is a Laurent polynomial, $k \, \in \, \mathbb{Z}$ and $A \, \in \, Mat_N \mathbb{C} $. The product in $\sqan$ is given by

\begin{equation}
(z^{m}f(T_q)A)(z^{k}g(T_q)B)=z^{m+k}f(q^{k}T_q)g(T_q)AB.
\end{equation}

Let $\sqn$ denote the Lie algebra obtained from $\sqan$ with the bracket given by the conmutator, namely:

\begin{equation}
[z^{m}f(T_q)A, z^{k}g(T_q)B]=z^{m+k}(f(q^{k}T_q)g(T_q)AB-f(T_q)g(q^m T_q)BA).
\end{equation}

Taking the trace form $tr_0(\sum_j c_jw^j)=c_0$, and denoting by $tr$ the usual trace in $Mat_M \mathbb{C}$, we obtain, by a general construction (cf. Sec. 1.3 in \cite{KR1}), the following $2$-cocylce in $\sqn$

\begin{equation}\label{eq:2cocycle}
\psi(z^mf(T_q)A,z^kg(T_q)B)= \delta_{m,-k} \, m \, tr_0(f(q^{-m}T_q)g(T_q)) \, tr(AB) ,
\end{equation}
where $r, \, s \, \in \, \mathbb{Z}$, $f, \, g \, \in \, \mathbb{C}[w,w^{-1}]$, $A, \, B \, \in \, Mat_N\mathbb{Z}$. Let

\begin{equation}
\widehat{\sqn} = \sqn^{\prime} \oplus \mathbb{C}C
\end{equation}
denote the central extension of $\sqn^{\prime}$ by a one-dimensional center $\mathbb{C}C$ corresponding to the two-cocycle $\psi$. The bracket in $\widehat{\sqn}$ is given by

\begin{align}
[z^{m}f(T_q)A, z^{k}g(T_q)B] & = z^{m+k}(f(q^{k}T_q)g(T_q)AB-f(T_q)g(q^m T_q)BA) \\ \nonumber
							& +\psi(z^{m}f(T_q)A, z^{k}g(T_q)B)C.
\end{align}

The elements $z^k T^m_q E_{ij} \, (k \, \in \, \mathbb{Z}, m \, \in \, \mathbb{Z}, i,j \, \in \, \{1,\cdots,N\})$ form a basis of $\sqn$. We define the \textit{weight} on $\sqn$ by

\begin{equation}\label{eq:weight}
wt z^k f(T_q)E_{ij}=kN+i-j.
\end{equation}

This gives  the \textit{principal} $\mathbb{Z}$ -gradation of $\sqan$, $\sqn$ and $\widehat{\sqn}$,

\begin{equation*}
\sqan= \oplus_{j \, \epsilon \, \mathbb{Z}} (\mathscr{S}_{q,N})_j, \quad  \widehat{\sqn}= \oplus_{j \, \epsilon \, \mathbb{Z}} (\widehat{\mathscr{S}}_{q,N})_j.
\end{equation*}
$\newline$
An $\textit{anti-involution}$ $\sigma$ of $\sqan$ is an involutive anti-automorphism of $\sqan$, ie,
$\sigma^2= Id$, $\sigma(\mathit{a}x+\mathit{b}y)=\mathit a \sigma(x)+\mathit b \sigma(y)$ and $\sigma(xy)=\sigma(y)\sigma(x)$, for all $a, \, b \, \in \, \mathbb{C}$ and $x, y \, \in \, \sqan$. From now on we will assume that $|q|\neq1$.

The following Corolary was proved in \cite{BB}.

\begin{corollary}\label{cor:antiinv}
Let $\sigma=\sigma_{A,B,c,r,N}$ be given by

\begin{equation*}
\sigma(E_{ii})=E_{N+1-i,N+1-i}
\end{equation*}

\begin{equation*}
\sigma(T_qE_{ii})=BT^{-1}_qE_{N+1-i,N+1-i}
\end{equation*}

\begin{equation}
\sigma(zE_{ii})=zAT^r_qE_{N+1-i,N+1-i}
\end{equation}

\begin{equation*}
\sigma(z^{-1}E_{ii})=A^{-1}q^{r}z^{-1}T^{-r}_qE_{N+1-i,N+1-i}
\end{equation*}

\begin{equation*}
\sigma(E_{ij})= \begin{cases} c_{i,j}E_{N+1-j,N+1-i} & \text{if} \quad i>j \\
								c^{-1}_{j,i}E_{N+1-j,N+1-i} & \text{if} \quad i<j
								\end{cases}
\end{equation*}
where $A$, $B$,  $c_{i,j}, \, r \, \in \, \mathbb{C}$, $A^2(Bq^{-1})^r=1$, $c_{i,j}$ verify the following relations
\begin{subequations}
\begin{equation} \label{eq:productosc}
c_{ij}=c_{i,i-1}c_{i-1,i-2}\cdots c_{j+1,j}
\end{equation}
\begin{equation}\label{eq:lamolesta}
\begin{cases} c_{i,j}c_{N+1-j,N+1-i}=1 & \text{if} \quad i \leq n \quad \textit{or} \quad j>n \\
			  c_{i,j}c^{-1}_{N+1-i,N+1-j}=\pm1 & \text{if} \quad i>n \quad \textit{and} \quad  j \leq n.
\end{cases}
\end{equation}
\end{subequations}
Then $\sigma=\sigma_{A,B,c,r,N}$ extends to an anti-involution on $\sqan$ which preserves the principal $\ZZ$-gradation.

\end{corollary}

\begin{remark}
For each $n<N$, a $\ZZ$-gradation preserving anti-involution can be constructed in a similar way. In \cite{BB} all anti-involutions of $\sqan$ preserving the $\ZZ$-gradation were classified.
\end{remark}

Let $\sqn^{A,B,c,r,N}$ denote the Lie subalgebra of $\sqn$ fixed by minus $\sigma_{A,B,c,r,N}$, namely

\begin{equation}
\sqn^{A,B,c,r,N}=\{ a \, \in \, \sqn | \sigma_{A,B,c,r,N}(a)=-a \},
\end{equation}
where $\sigma_{A,B,c,r,N}$ is the anti-involution given by Corolary \ref{cor:antiinv}.

\begin{lemma}
The Lie algebras $\sqn^{A,B,c,r,N}$ for arbitrary choices of $  A, \, B \text{ and } c$
are isomorphic to $\sqn^{\epsilon, q, \mathbf{1}, r, N}$, where $\epsilon$ is $1$ or $-1$, and $\mathbf{1}$ is the matrix $c$ with $c_i=1$ except for the fixed points that are $1$ or $-1$, which keep their sign.
\end{lemma}

Thus, the anti-involution is of the following form:
\begin{equation}\label{eq:sigma}
\sigma_{\epsilon,r,N}(z^kh(T_q)E_{i,j})=(\epsilon)^{k}q^{k(k-1)r/2}z^k h(q^{1-k}T^{-1}_q)T^{kr}_qE_{N+i-j,N+1-i}.
\end{equation}
where $\epsilon = \pm 1$, $r \, \in \, \mathbb{C^{\times}}$. For simplicity, denote $\sqnsnn$ the Lie subalgebras of $\sqn$ fixed by minus $\sigma_{\epsilon,r,N}$.

We will denote

\begin{equation*}
\delta_{m,\text{even}}= \begin{cases} 1 & \text{ if } \,  m \, \text{ is even} \\
									  0 & \text{ otherwise}
									  \end{cases}.
\end{equation*}
$\sqnsnn$ inherits a $\mathbb{Z}$-gradation from $\sqn$ since $\sigma$ preserves the principal $\mathbb{Z}-$gradation of $\sqan$. Thus $\sqnsnn = \oplus_{j \, \in \, \mathbb{Z}} (\sqnsnn)_j$. We can now give a description of $(\sqnsnn)_j$. By the division algorithm, let $j=kN+p$ with $0 \leq p \leq N-1$ . Thus,

If $p \neq 0$
\begin{align*}
(\sqnsnn)_j = &  \{ z^k(q^{(k-1)/2}T_q)^{rk/2}(f(q^{(k-1)/2}T_q)E_{i,i-p}\\
&-(\epsilon)^kf((q^{(k-1)/2}T_q)^{-1})E_{N+1-i+p,N+1-i}) | f(w) \, \in \, \mathbb{C}[w, w^{-1}], \, 1+p \leq i \leq N, \\
& \, i \neq (N+1+p)/2 \} \\
& \bigcup \delta_{N+p, \text{odd}} \{ z^{k}(q^{(k-1)/2}T_q)^{rk/2}g(q^{(k-1)/2}T_q)E_{(N+1+p)/2,(N+1-p)/2}  | g(w) \, \in \, \mathbb{C}[w, w^{-1}]^{\epsilon,k} \} \\
& \bigcup \{ z^{k+1}(q^{k/2}T_q)^{r(k+1)/2}(h(q^{k/2}T_q)E_{i,N-p+i}\\
&-(\epsilon)^{k+1}h(q^{-k/2}T^{-1}_q)E_{p+1-i,N+1-i} | h(w) \, \in \, \mathbb{C}[w, w^{-1}], \,  1 \leq i \leq p , \, i \neq (1+p)/2 \} \\
& \bigcup \delta_{p, \text{odd}} \{ z^{k+1}(q^{k/2}T_q)^{r(k+1)/2} \tilde{g}(q^{k/2}T_q)E_{(p+1)/2,(2N+1-p)/2} | \tilde{g}(w) \, \in \, \mathbb{C}[w, w^{-1}]^{\epsilon,k+1} \},
\end{align*}
and for $p=0$
\begin{align*}
(\sqnsnn)_j = &  \{ z^k(q^{(k-1)/2}T_q)^{rk/2}(f(q^{(k-1)/2}T_q)E_{i,i}\\
&-(\epsilon)^kf((q^{(k-1)/2}T_q)^{-1})E_{N+1-i,N+1-i}) | f(w) \, \in \, \mathbb{C}[w, w^{-1}], \, 1 \leq i \leq [N/2] \} \\
& \bigcup \delta_{N, \text{odd}} \{ z^{-k}(q^{(-k-1)/2}T_q)^{-rk/2}g(q^{(-k-1)/2}T_q)E_{(N+1)/2,(N+1)/2}  | g(w) \, \in \, \mathbb{C}[w, w^{-1}]^{\epsilon,k} \} \\
\end{align*}
where $\mathbb{C}[w, w^{-1}]^{\epsilon,k}$ denotes the set of Laurent polynomials such that $f(w^{-1})=-(\epsilon)^kf(w)$.

We denote again $\psi$ the restriction of the $2$-cocycle in \eqref{eq:2cocycle} to $\sqnsnn$.
%
Denote by $\sqnh$ the central extension of $\sqns$ by $\mathbb{C}C$ corresponding to the $2$-cocycle $\psi$. $\sqnh$ is a Lie subalgebra of $\widehat{\sqn}$ by definition.


\section{Parabolic subalgebras of $\sqnsnn$}
In order to characterize the quasifiniteness of the highest weight modules (HWMs) of $\sqnh$ we will study the structure of its parabolic subalgebras and apply general results for quasifinite representations of $\ZZ$-graded Lie algebras obtained in \cite{KL}. We refer to \cite{KL} for proofs and details.
Let $\mathfrak{g}$ be a $\mathbb{Z}$-graded Lie algebra over $\mathbb{C}$,

\begin{equation*}
\mathfrak{g} = \bigoplus_{j \, \in \, \mathbb{Z}} \mathfrak{g}_j, \quad [\mathfrak{g}_i, \mathfrak{g}_j]\subset \mathfrak{g}_{i+j},
\end{equation*}
where $\mathfrak{g}_i$ is not necessarily of finite dimension. Let $\mathfrak{g}_{\pm}=\oplus_{j>0} \mathfrak{g}_{\pm j}$. A subalgebra $\mathfrak{p}$ of $\mathfrak{g}$ is called \textit{parabolic} if it contains $\mathfrak{g}_0 \oplus \mathfrak{g}_+$ as a proper subalgebra, that is

\begin{equation*}
\mathfrak{p} = \bigoplus_{j \, \in \, \mathbb{Z}} \mathfrak{p}_j, \quad \text{ where } \mathfrak{p}_j=\mathfrak{g}_j \text{ for } j \geq 0, \text{ and } \mathfrak{p}_j \neq 0 \text{ for some } j<0.
\end{equation*}

Following \cite{KL}, we assume the following properties of $\mathfrak{g}$:

\begin{enumerate}
\item[(P1)] $\mathfrak{g}_0$ is commutative,
\item[(P2)] if $a \, \epsilon \, \mathfrak{g}_{-k}$ ($k>0$) and $[a,\mathfrak{g}_1]=0$, then $a=0$.
\end{enumerate}

Given $a \, \in \, \mathfrak{g_{-1}}$, $a \neq 0$, we define $\mathfrak{p}^a= \oplus_{j \, \epsilon \, \mathbb{Z}} \mathfrak{p}_j$, where $\mathfrak{p}_j^a=\mathfrak{g}_j$ for all $j \geq 0$, and

\begin{equation*}
\mathfrak{p}_{-1}^a= \sum [ \dots [[a, \mathfrak{g}_0], \mathfrak{g}_0], \dots] , \quad \quad \mathfrak{p}_{-k-1}^a = [\mathfrak{p}_{-1}^a, \mathfrak{p}_{-k}^a].
\end{equation*}

\begin{lemma}\label{eq:lemaremarkvane}
\begin{itemize}
\item[(a)]
For any parabolic subalgebra $\mathfrak{p}$ of $\mathfrak{g}$, $\mathfrak{p}_{-k} \neq 0, \, k>0$, implies \\ $\mathfrak{p}_{-k+1} \neq 0$.

\item[(b)] $\mathfrak{p}^a$ is the minimal parabolic subalgebra containing $a$.

\item[(c)] $\mathfrak{g}^a_0 := [\mathfrak{p}^a,\mathfrak{p}^a] \cap \mathfrak{g}_0 = [a, \mathfrak{g}_1]$
\end{itemize}
\end{lemma}

\begin{proof}
Cf. \cite{KL} Lemmas 2.1 and 2.2.
\end{proof}

In \cite{BKLY}, for the case of the central extension of the Lie algebra of matrix differential operators on the circle, the existance of some parabolic subalgebras $\mathfrak{p}$ such that $\mathfrak{p}_j = 0$ for $j\gg 0$ was observed. Having in mind that example, they give the following definition.

\begin{definition}
\begin{itemize}
\item[(a)] A parabolic subalgebra $\mathfrak{p}$ is called $\mathit{nondegenerate}$ if $\mathfrak{p}_{-j}$ has finite codimension in $\mathfrak{g}_{-j}$ , for all $j>0$.
\item[(b)] An element $a \, \in \, \mathfrak{g}_{-1}$ is called $\mathit{nondegenerate}$ if $\mathfrak{p}^a$ is nondegenerate.
\end{itemize}
\end{definition}

We will also require the following condition on $\mathfrak{g}$.

\begin{enumerate}
\item[(P3)] If $\mathfrak{p}$ is a nondegenerate parabolic subalgebra of $\mathfrak{g}$, then there exists an nondegenerate element $a$ such that $\mathfrak{p}^a \subseteq \mathfrak{p}$.
\end{enumerate}

Now take a parabolic subalgebra $\mathfrak{p}$ of $\sqnh$. Observe that for each $j \, \in \, \mathbb{N}$, $j=kN+p$ with $0 \leq p \leq N-1$, we have
\begin{align}\label{pj}
&\mathfrak{p}_{-j} = \{ z^{-k}(q^{(-k-1)/2}T_q)^{-rk/2}(f(q^{(-k-1)/2}T_q)E_{i,i+p} \\ \nonumber
&-(\epsilon)^kf(q^{(k+1)/2}T_q^{-1})E_{N+1-i-p,N+1-i}) | f(w) \, \in \, I^i_{-j}, \, 1 \leq i \leq N-p, \, i \neq (N+1-p)/2 \} \\ \nonumber
& \bigcup \delta_{N-p, \text{odd}} \{ z^{-k}(q^{(-k-1)/2}T_q)^{-rk/2}g(q^{(-k-1)/2}T_q)E_{(N+1-p)/2,(N+1+p)/2}  | g(w) \, \in \, I^{(N+1-p)/2}_{-j} \} \\ \nonumber
&\bigcup \{ z^{-k-1}(q^{-1-k/2}T_q)^{-r(k+1)/2}(h(q^{-1-k/2}T_q)E_{i,i-N+p} \\ \nonumber
&-(\epsilon)^{(k+1)}h((q^{-1-k/2}T_q)^{-1})E_{2N+1-i-p,N+1-i}) | h(w) \, \in \, I^i_{-j}, \, N+1-p \leq i \leq N, \, i \neq (2N+1-p)/2  \} \\ \nonumber
& \bigcup \delta_{p, \text{odd}} \{ z^{-k-1}(q^{-1-k/2}T_q)^{-r(k+1)/2} \tilde{g}(q^{-1-k/2}T_q)E_{(2N+1-p)/2,(1+p)/2} | \tilde{g}(w) \, \in \, I^{(2N+1-p)/2}_{-j} \},
\end{align}
where $I^i_{-j}$ is a subspace of $\mathbb{C}[w, w^{-1}]$, $I^{(N+1-p)/2}_{-j}$ is a subspace of $\mathbb{C}[w, w^{-1}]^{\epsilon,k}$ and $I^{(2N+1-p)/2}_{-j}$ is a subspace of $\mathbb{C}[w, w^{-1}]^{\epsilon,k+1}$.

Let us check conditions (P1),(P2) and (P3) for $\sqnh$.

Observe that (P1) is immediate from the definition of $(\sqnh)_0$. (P2) follows from computing the following bracket
\begin{align*}
[z^{l}(q^{(l-1)/2}T_q)^{lr/2}(f(T_q)E_{i,j}-&(\epsilon)^lf(T^{-1}_q)E_{N+1-j,N+1-i}), E_{j,j-1}-E_{N+2-j,N+1-j}]
\end{align*}
and the particular case
\begin{align*}
[z^{l}(q^{(l-1)/2}T_q)^{lr/2}(f(T_q)E_{N/2,N/2}-&(\epsilon)^lf(T^{-1}_q)E_{N/2+1,N/2+1}), E_{N/2,N/2-1}-E_{N/2+2,N/2+1}].
\end{align*}

%
%

$ $\newline
To prove (P3), let $f(w), \, g(w)$ be Laurent polynomials in the variable $w$ with $f \, \in I^i_{-j}$, and let $\mathfrak{p}_{-j}$ with $j = kN+p$ as in \eqref{pj}. Let us first consider $1 \leq i \leq N-p$. If $p =0$, suppose $i \neq (N+1)/2$. We compute the following bracket
\begin{align*}
[z^{-k}(q^{(-k-1)/2}T_q)^{-rk/2}(f(q^{(-k-1)/2}T_q)E_{i,i}-&(\epsilon)^kf(q^{(k+1)/2}T_q^{-1})E_{N+1-i,N+1-i}), \\
& g(q^{-1/2}T_q)E_{i,i}-g(q^{1/2}T^{-1}_q)E_{N+1-i,N+1-i}]
\end{align*}
So, $I^i_{-j}$ satisfies
\begin{equation}\label{eq:IjIdeal}
A_{j}I^i_{-j} \subseteq I^i_{-j},
\end{equation}
where $A_j = \{ g(q^{k/2}w)-g(q^{-k/2}w): g(w) \, \in \, \mathbb{C}[w, w^{-1}] \}$.

If $p \neq 0$, suppose that $i \neq N+1-i$. Computing the following bracket
\begin{align*}
[z^{-k}(q^{(-k-1)/2}T_q)^{-rk/2}(f(q^{(-k-1)/2}T_q)E_{i,i+p}-&(\epsilon)^kf(q^{(k+1)/2}T_q^{-1})E_{N+1-i-p,N+1-i}), \\
& g(q^{-1/2}T_q)E_{i,i}-g(q^{1/2}T^{-1}_q)E_{N+1-i,N+1-i}],
\end{align*}
we see that $I^i_{-j}$ satisfies \eqref{eq:IjIdeal} for $A_j = \{ g(q^{k/2}w): g(w) \, \in \, \mathbb{C}[w, w^{-1}] \}$.

Now, if $N+1-p \leq i \leq N$, we see by computing
\begin{align*}
[z^{-k-1}(q^{-1-k/2}T_q)^{-r(k+1)/2}(f(q^{-1-k/2}T_q)E_{i,i-N+p}-&(\epsilon)^{k+1}f(q^{1+k/2}T_q^{-1})E_{2N+1-i-p,N+1-i}), \\
& g(q^{-1/2}T_q)E_{i,i}-g(q^{1/2}T^{-1}_q)E_{N+1-i,N+1-i}]
\end{align*}
that $I^i_{-j}$ also satisfies \eqref{eq:IjIdeal}, this time for $A_j = \{ g(q^{(-1-k)/2}w): g(w) \, \in \, \mathbb{C}[w, w^{-1}] \}$.

Analogous results can be obtained if $N-p$ is odd for $I^{(N+1-p)/2}_{-j}$ by computing
\begin{align*}
[z^{-k}(q^{(-k-1)/2}T_q)^{-kr/2}&(f(q^{(-k-1)/2}T_q)E_{(N+1-p)/2,(N+1+p)/2}, \\
& g(q^{-1/2}T_q)E_{i,i}-g(q^{1/2}T^{-1}_q)E_{N+1-i,N+1-i}],
\end{align*}
and if $p$ is odd for $I^{(2N+1-p)/2}_{-j}$, computing
\begin{align*}
[z^{-k-1}(q^{-1-k/2}T_q)^{-r(k+1)/2}&(f(q^{-1-k/2}T_q)E_{(2N+1-p)/2,(1+p)/2}, \\
& g(q^{-1/2}T_q)E_{i,i}-g(q^{1/2}T^{-1}_q)E_{N+1-i,N+1-i}].
\end{align*}

Thus, since $\CC[w, w^{-1}]$ is a principal ideal domain, we have proven the following

\begin{lemma}\label{eq:LemaIj}
For $j>0,$
\begin{itemize}
\item[(a)]  $I^i_{-j}$, $I^{(N+1-p)/2}_{-j}$ and $I^{(2N+1-p)/2}_{-j}$ are ideals;
\item[(b)] if $I^i_{-j} \neq 0$, $I^{(N+1-p)/2}_{-j} \neq 0$ and $I^{(2N+1-p)/2}_{-j} \neq 0$, then they have finite codimension in $\mathbb{C}[w, w^{-1}]$.
\end{itemize}
\end{lemma}

\medskip

Let $[k]$ denote the integer part of a number $k$. Now we have the following important Proposition.

\begin{proposition}\label{eq:ProposicionSobreD}
\begin{itemize}
\item[(a)] Any nonzero element $d \, \in \, (\sqnh)_{-1}$ is nondegenerate.
\item[(b)] Any parabolic subalgebra of $\sqnh$ is nondegenerate.
\item[(c)] Let  $d \, \in \, (\sqnh)_{-1}$,
\begin{align*}
d= & \sum^{[N/2]+\delta_{N, even}}_{i=1} f_i(q^{-1/2}T_q)E_{i,i+1}-f_i(q^{1/2}T^{-1}_q)E_{N-i, N+1-i} \\
& + \delta_{N, even} g(q^{-1/2}T_q)E_{N/2,N/2+1}+ z^{-1}(q^{-1}T_q)^{-r/2}h(q^{-1}T_q)E_{N,1},
\end{align*}
where $f_i(w), \, g(w)$ and $h(w)$ are Laurent polynomials such that $g(w^{-1})=-g(w)$ and $h(w^{-1})=-\epsilon h(w)$. Then

\begin{align*}
(\sqnh)^{d}_{0} : & = [(\sqnh)_{1}, d] \\
& = \hbox{span} \{ f_{k-1}(q^{-1/2}T_q)(q^{-1/2}T_q)^l(E_{k-1,k-1}-E_{k,k}) \\
& + f_{k-1}(q^{1/2}T^{-1}_q)(q^{-1/2}T_q)^{-l}(E_{N-k+1,N-k+1}-E_{N+2-k,N+2-k}): \\
& 2 \leq k \leq [N/2]+ \delta_{N, odd}, l \, \in \, \mathbb{Z}_{\geq 0} \} \\
& \cup \delta_{N, even} \{ g(q^{-1/2}T_q)((q^{-1/2}T_q)^n-(q^{-1/2}T_q)^{-n})(E_{N/2,N/2}-E_{N/2+1,N/2+1}): \\
& n \, \in \, \mathbb{Z}_{\geq 0}, \, g \, \in \, \mathbb{C}[w, w^{-1}]^{\epsilon,0} \} \\
& \cup \{ h(T_q)(T^m_q-\epsilon T^{-m}_q)E_{N,N}-h(q^{-1}T_q)((q^{-1}T_q)^m-\epsilon (q^{-1}T_q)^{-m})E_{1,1} \\
& +tr_0(\epsilon h(q^{1/2}w^{-1})(w^m-\epsilon w^{-m}))C: m \, \in \, \mathbb{Z}, \, h \, \in \,  \mathbb{C}[w, w^{-1}]^{\epsilon,1}  \}.
\end{align*}
\end{itemize}
\end{proposition}

\begin{proof}
Let $0 \neq d \, \in \, (\sqnh)_{-1}$, by Lemma \eqref{eq:lemaremarkvane}, part (a), $\mathfrak{p}^d_{-j} \neq 0$ for all $j \geq 1$. So, by  Lemma \eqref{eq:LemaIj} part (b), part (a) follows. Let $\mathfrak{p}$ be any parabolic subalgebra  of $\sqnh$, using Lemma 2.1 and 2.2 in \cite{KL}, we get $\mathfrak{p}_{-1} \neq 0$. Then using (a) and $\mathfrak{p}^d \subseteq \mathfrak{p}$ (for any nonzero $d \, \in \, \mathfrak{p}_{-1}$) we obtain (b). Finally, part (c) follows by computing the commutators $[d,a]$ with $a = (q^{-1/2}T_q)^lE_{k,k-1}-(q^{-1/2}T_q)^{-l}E_{N+2-k,N+1-k} $ with $2 \leq k \leq [N/2]+ \delta_{N, odd}$; $a= \delta_{N, even}((q^{-1/2}T_q)^n-(q^{-1/2}T_q)^{-n})E_{N/2+1,N/2}$ and $a= z T^{r/2}_q(T^m_q-\epsilon T^{-m}_q)E_{1,N}$, with $l, \, n, \, m \, \in \, \mathbb{Z}_{\geq 0}$.
\end{proof}

Summarizing, we have proven that the following properties are satisfied by $\sqnh$:
\begin{itemize}
\item[(P1)] $(\sqnh)_0$ is commutative;
\item[(P2)] if $a \, \in \, (\sqnh)_{-j} \, (j>0)$ and $[a,(\sqnh)_1]=0$, then $a=0$;
\item[(P3)] if $\mathfrak{p}$ is a nondegenerate parabolic subalgebra of $\sqnh$, then there exists a nondegenerate element $a$, such that $\mathfrak{p}^a \subseteq \mathfrak{p}.$
\end{itemize}

Observe that (P3) follows from Proposition \eqref{eq:ProposicionSobreD}, parts (a) and (b).


\section{Characterization of quasifinite highest weight modules of $\sqnh$}
Now, we begin our study of quasifinite representations over the lie algebras $\sqnh$. Let $\mathfrak{g}$ be a Lie algebra. For a Lie algebra $\mathfrak{g}$, a $\mathfrak{g}$-module $V$ is called $\mathbb{Z}$-graded if $V= \oplus_{j \, \in \, \mathbb{Z}} V_j$ and $\mathfrak{g}_i V_j \subset V_{i+j}$. A $\mathbb{Z}$-graded $\mathfrak{g}$-module $V$ is called \textit{quasifinite} if $\hbox{dim} V_j < \infty$ for all $j$.

Given $\lambda \, \in \, \mathfrak{g}^{\ast}_0$, a \textit{highest weight module} is a $\mathbb{Z}$-graded $\mathfrak{g}$-module $V(\mathfrak{g},\lambda)$ generated by a heighest weight vector $v_{\lambda} \, \in \, V(\mathfrak{g}, \lambda)$ which satisfies
\begin{equation*}
hv_{\lambda} = \lambda(h)v_{\lambda} \quad (h \, \in \, \mathfrak{g}_0), \quad \quad \mathfrak{g}_+v_{\lambda}=0.
\end{equation*}
A nonzero vector $v \, \in \, V(\mathfrak{g}, \lambda)$ is called $\mathit{singular}$ if $\mathfrak{g}_+v_{\lambda}=0.$ The \textit{Verma module} over $\mathfrak{g}$ is defined as usual:

\begin{equation*}
M(\mathfrak{g},\lambda) = \mathcal{U}(\mathfrak{g}) \otimes_{\mathcal{U}(\mathfrak{g}_0 \oplus\mathfrak{g}_+)} \mathbb{C}_{\lambda},
\end{equation*}
where $\mathbb{C}_{\lambda}$ is the one-dimensional $(\mathfrak{g}_0 \oplus\mathfrak{g}_+)$-module given by $h \mapsto \lambda(h)$ if $h \, \in \, \mathfrak{g}_0$, $\mathfrak{g}_+ \mapsto 0$, and under the action of $\mathfrak{g}$ is induced by the left multiplication in $\mathcal{U}(\mathfrak{g})$. Here and further $\mathcal{U}\mathfrak{(g)}$ stands for the universal enveloping algebra of the Lie algebra $\mathfrak{g}$. Any highest-weight module $V(\mathfrak{g}, \lambda)$ is a quotient module of $M(\mathfrak{g}, \lambda)$. The irreducible module $L(\mathfrak{g}, \lambda)$ is the quotient of $M(\mathfrak{g}, \lambda)$ by the maximal proper graded module. We shall write $M(\lambda)$ and $L(\lambda)$ in place of $M(\mathfrak{g}, \lambda)$ and $L(\mathfrak{g}, \lambda)$ if no ambiguity may arise.

Consider a parabolic subalgebra $\mathfrak{p}= \oplus_{j \, \in \, \mathbb{Z}} \mathfrak{p}_j$ of $\mathfrak{g}$ and let $\lambda \, \in \, \mathfrak{g}^{\ast}_0$ be such that $\lambda|_{\mathfrak{g}_0 \cap [\mathfrak{p},\mathfrak{p}]}=0$. Then the $(\mathfrak{g}_0 \oplus \mathfrak{g}_+)$-module  $\mathbb{C}_{\lambda}$ extends to a $\mathfrak{p}$-module by letting $\mathfrak{p}_j$ act as $0$ for $j <0$, and we may construct the highest-weight module

\begin{equation*}
M(\mathfrak{g}, \mathfrak{p}, \lambda)=\mathcal{U} \otimes_{\mathcal{U}(\mathfrak{p})} \mathbb{C}_{\lambda}
\end{equation*}
called the \textit{generalized Verma module}. Clearly all these heighest weight modules are graded.

From now on we will consider $\lambda \, \in \, \widehat{\mathfrak{g}}_0^{\ast}$. By Theorem 2.5 in \cite{K}, we have the following.

\begin{theorem}\label{eq:teoremacuasifinito}
The following conditions on $\lambda \, \in \, \mathfrak{g}^{\ast}_0$ are equivalent:
\begin{enumerate}
\item $M(\lambda)$ contains a singular vector $a . v_{\lambda}$ in $M(\lambda)_{-1}$ where $a$ is nondegenerate;
\item there exists a nondegenerate element $a \, \in \, \mathfrak{g}_{-1}$, such that $\lambda([\mathfrak{g}_1,a])=0$;
\item $L(\lambda)$ is quasifinite;
\item there exists a nondegenerate element  $a \, \in \, \mathfrak{g}_{-1}$, such that $L(\lambda)$ is the irreducible quotient of the generalized Verma module $M(\mathfrak{g}, \mathfrak{p}^a,\lambda)$.
\end{enumerate}
\end{theorem}

$ $\newline

Consider $\widehat{\mathfrak{g}}=\sqnh$. A functional $\lambda \, \in \, (\sqnh)^{\ast}_0$ is described by its \textit{labels},
\begin{align*}
\bigtriangleup_{i,l} =  \lambda ((q^{-1/2}T_q)^lE_{i,i}-(q^{-1/2}T_q)^{-l}E_{N+1-i,N+1-i}), \\
\bigtriangleup_{N,l} =  \lambda ((T_q^l+T_q^{-l})E_{N,N}-((q^{-1}T_q)^l+(q^{-1}T_q)^{-l})E_{1,1})
\end{align*}
with $l \, \in \, \mathbb{Z}_{\geq 0}$, $1 < i \leq [N/2]+\delta_{N,even}$ and the \textit{central charge} $c=\lambda(C)$. We shall consider the generating series

\begin{equation}
\bigtriangleup_i(x)= \sum_{l \in \mathbb{Z}} x^{-l} \bigtriangleup_{i,l} \quad 1 < i \leq [N/2]+\delta_{N,even} \quad \text{ and } \quad \bigtriangleup_N(x)= \sum_{l \in \mathbb{Z}} x^{-l} \bigtriangleup_{N,l}.
\end{equation}

Recall that a \textit{quasipolynomial} is a linear combination of functions of the form $p(x)q^{\alpha x}$, where $p(x)$ is a polynomial and $\alpha \, \in \, \mathbb{C}$. That is, it satisfies a nontrivial linear differential equation with constant coefficients. We also have the following well-known proposition.

\begin{proposition}\label{eq:proposicionquasi}
Given a quasipolynomial $P$, and a polynomial $B(x)= \prod_i(x-A_i)$, take $b(x)=\prod_i(x-a_i)$ where $a_i=e^{A_i}$, then $b(x)(\sum_{n \, \in \, \mathbb{Z}}P(n)x^{-n})=0$ if and only if $B(d/dx)P(x)=0$.
\end{proposition}
If the polynomial $B$ is even we call $P$ an \textit{even quasipolynomial}. As a result, one has the following characterization of quasifinite heighest weight modules over $\widehat{\mathfrak{g}}$.


\begin{theorem}\label{eq:caracterizacion}
A $\sqnh$-module $L(\lambda)$ is quasifinite if and only if one of the following conditions holds:
\begin{enumerate}

\item There exist monic polynomials $b_1(w), \cdots , b_{[N/2]-\delta_{N, even} }(x), b^{\epsilon}_N(w)$ such that
\begin{align}
b_i(x)(\bigtriangleup_{i+1}(x)- & \bigtriangleup_{i}(x))=0 \quad \text{ for } \quad 1 < i \leq [N/2]-\delta_{N, even}  \quad \text{ and } \label{eq:deltagral} \\
& b^{\epsilon}_N(x)(\bigtriangleup_N(x)+2c)=0 \label{eq:deltacasoN}
\end{align}
Moreover, if $N$ is even there exists a monic polynomial $b_{N/2}(x)$ such that
\begin{equation}
b_{N/2}(x)(\bigtriangleup_{1+N/2}(x)-\bigtriangleup_{N/2}(x))=0
\end{equation}
\item There exist quasipolynomials $P_i$ and even quasipolynomials $P^{\epsilon}_N$ such that $( n \in \NN)$
\begin{align}
P_{i}(n)=\bigtriangleup_{i,n}-\bigtriangleup_{i+1,n} \quad \text{ for } \quad  1 < i \leq [N/2]-\delta_{N, even} \quad  \text{ and } \label{eq:pigral} \\
P^{\epsilon}_N(n)=\bigtriangleup_{N,n} \quad \text{ for } \quad  n \neq 0 \quad \text{ and } \quad P^{\epsilon}_N(0)=-2c. \label{eq:pnepsilon}
\end{align}
Moreover, if $N$ is even there exists an even quasipolynomial $P_{N/2}$ such that
\begin{equation} \label{eq:pnsobre2}
P_{N/2}(n)=\bigtriangleup_{N/2,n}-\bigtriangleup_{N/2+1,n}.
\end{equation}

\end{enumerate}
\end{theorem}


\begin{proof}
From Proposition \ref{eq:ProposicionSobreD} part (c) and Theorem \ref{eq:teoremacuasifinito} part (b), we have that $L(\lambda)$ is quasifinite if and only if there exist (monic) Laurent polynomials

\begin{equation*}
h^{\epsilon}(w)= \sum ^{p}_{t=0} c_t(w^t-\epsilon w^{-t}), \quad g(w)= \sum ^{u}_{s=0} d_s (w^s-w^{-s}), \quad f_i(w)= \sum ^{m_i}_{v=-m_i} a_{i,v} w^v
\end{equation*}
for $1 < i \leq [N/2]- \delta_{N,even}$, such that for each $l, \, n, \, m \, \epsilon \, \mathbb{Z}_{\geq 0}$, we have
\begin{align*}
& \lambda(f_{k-1}(q^{-1/2}T_q)(q^{-1/2}T_q)^l(E_{k-1,k-1}-E_{k,k}) \\
& + f_{k-1}(q^{1/2}T^{-1}_q)(q^{-1/2}T_q)^{-l}(E_{N-k+1,N-k+1}-E_{N+2-k,N+2-k}))=0
\end{align*}
with $1 < k \leq [N/2]- \delta_{N,even}$,
\begin{align*}
\lambda(& h(T_q)(T^m_q-\epsilon T^{-m}_q)E_{N,N}-h(q^{-1}T_q)((q^{-1}T_q)^m-\epsilon (q^{-1}T_q)^{-m})E_{1,1} \\
&  +tr_0(\epsilon h(q^{1/2}w^{-1})(w^m-\epsilon w^{-m}))C)=0,
\end{align*}
and
\begin{equation*}
\delta_{N, even}\lambda(g(q^{-1/2}T_q)((q^{-1/2}T_q)^n-(q^{-1/2}T_q)^{-n})(E_{N/2,N/2}-E_{N/2+1,N/2+1}))=0.
\end{equation*}
These conditions can be rewritten as follows:

\begin{equation}\label{eq:ecuacionconaies}
0=\sum ^{m_i}_{v=-m_i} a_{i,v}(\bigtriangleup_{i,v+l}-\bigtriangleup_{i+1,v+l})
\end{equation}
for all $1 < i \leq [N/2]- \delta_{N,even}$ and $l \, \in \, \mathbb{Z}_{\geq 0}$, and
\begin{equation}\label{eq:ecuaciondeh}
0=\sum^{p}_{t=0}c_t(\bigtriangleup_{N,t+m}-\epsilon \bigtriangleup_{N,t-m})+tr_0(\epsilon h(q^{1/2}w^{-1})(w^m-\epsilon w^{-m}))C)=0
\end{equation}
with  $m \, \in \, \NN$. Finally, if $N$ is even,
\begin{equation}\label{eq:ecuacionN/2}
0=\sum^{u}_{s=0} d_s  (\bigtriangleup_{N/2,s+n}-\bigtriangleup_{N/2,-s+n}+\bigtriangleup_{1+N/2,-s+n}-\bigtriangleup_{1+N/2,s+n})
\end{equation}
with  $n \, \in \, \mathbb{Z}_{\geq 0}$. Let

\begin{equation*}
F_{k}(x)=\bigtriangleup_k(x)-\bigtriangleup_{k+1}(x)
\end{equation*}
for  $1 < k \leq [N/2]- \delta_{N,even}$.

Let us first analyze \eqref{eq:ecuacionconaies}. Multiplying both sides by $x^{-l}$ and adding over $l \in \mathbb{Z}$, we get
\begin{equation}
0=\sum ^{m_i}_{v=-m_i} a_{i,v} x^v F_i(x)=f_i(x)F_i(x).
\end{equation}
We construct $\tilde{b}_i(x)=x^{m_i} f_{i}(x) \in \CC[x]$. The equivalence of (1) and (2) for this case follows from the fact that \eqref{eq:deltagral} holds since it also holds multiplying both sides of this formula by $x^{m_i}$ with $m_i \geq 0$. Due to Proposition \ref{eq:proposicionquasi}, the existence of the qualispolynomials $P_i(x)$ for $1 < i \leq [N/2]-\delta_{N, even}$ is clear.

Let us now study \eqref{eq:ecuaciondeh}. Making use of the definition of $tr_0$ given in Section \ref{sec:sqn} and the fact that $\bigtriangleup_{N,l}= \bigtriangleup_{N,-l}$, we get

\begin{equation*}
0=\sum^{p}_{t=0}c_t(\bigtriangleup_{N,t+m}-\epsilon \bigtriangleup_{N,t-m})-2\epsilon c_m c
\end{equation*}

Multiplying both sides by $x^m-\epsilon x^{-m}$ and adding over $m \in \ZZ_{\geq 0}$, we obtain

\begin{equation*}
0=\sum^{p}_{t=0}c_t(x^{-t}-\epsilon x^{t})\bigtriangleup_N(x)-\sum_{m \in \mathbb{Z}}(x^m-\epsilon x^{-m}) (2\epsilon c_m c)=-\epsilon h^{\epsilon}(x)(\bigtriangleup_{N}(x)+2c).
\end{equation*}

Once again, \eqref{eq:deltacasoN} holds since it also holds multiplying both sides of this formula by $x^p$ with $p \geq 0$. Now, $b^{\epsilon}(x)= x^ph^{\epsilon}(x) \in \CC[x]$. Since $h^{\epsilon}(x^{-1})=-\epsilon h^{\epsilon}(x)$ it is easy to see that if $\alpha \neq 0$ is a root of $b^{\epsilon}(x)$, then $1/\alpha$ is also a root of $b^{\epsilon}(x)$. Now we can apply Proposition \ref{eq:proposicionquasi} and due to the relationship between the roots of $B$ and $b$ in this proposition it follows that the $B^{\epsilon}(x)$ corresponding to our $b^{\epsilon}(x)$ is an even polynomial. This implies that the quasipolynomial $P^{\epsilon}_N(x)$ such that $P^{\epsilon}_N(n)=\bigtriangleup_{N,n}$ for $n\neq 0$ and $P(0)=2c$ is even, finishing the proof for this case.

\medskip

Finally, let us analyze \ref{eq:ecuacionN/2} for the case $N$ even. Proceding similarly as with the previous equation, we multiply by $(x^{n}-x^{-n})$ and add over $n \, \in \, \mathbb{Z}_{\geq 0}$. Using the fact that $\bigtriangleup_{1+N/2,l}=-\bigtriangleup_{N/2,-l}$ we obtain

\begin{equation*}
0=\sum^u_{s=0}d_s(x^{s}-x^{-s})(\bigtriangleup_{N/2}(x)-\bigtriangleup_{N/2+1}(x))= g(x) F_{N/2}(x).
\end{equation*}

Now $\hat{b}_{N/2}(x) = x^{u}g(x) \in \CC[x]$. Making use once again of Proposition \ref{eq:proposicionquasi} we prove that $P_{N/2}(x)$ such that $P_{N/2}(n)=\bigtriangleup_{N/2,n}-\bigtriangleup_{N/2+1,n}$ for $n \in \mathbb{Z}$ is an even quasipolynomial.
\end{proof}


Given a quasifinite irreducible highest weight $\sqnh$-module $V$ by Theorem \ref{eq:caracterizacion}, we have that there either exist a quasipolynomials $P_i(x)$ (for $1 \leq i \leq [N/2]-\delta_{N,even}$) satisfying \eqref{eq:pigral}, an even quasipolynomials $P^{\epsilon}_N(x)$ verifying \eqref{eq:pnepsilon}, and if $N$ is even, $P_{N/2}(x)$ satisfying \eqref{eq:pnsobre2}. We will write
\begin{align} \label{eq:formulaspn}
&P_{i}(x)=\sum_{e \in \CC} p_{e,i}(x)q^{e_ix}, \\ \nonumber
&P^{\epsilon}_N(x)=\sum_{j \in \CC}p^{\epsilon}_{j,N}(x)\cosh_q(e_j^+x)+\sum_{j \in \CC}q^{\epsilon}_{j,N}(x)\sinh_q(e_j^-x) \quad \text{ and } \\ \nonumber
&P_{N/2}(x)=\sum_{j \in \CC}p_{j,N/2}(x)\cosh_q(e_j^+x)+\sum_{j \in \CC}q_{j,N/2}(x)\sinh_q(e_j^-x),
\end{align}	
with $p_{j,N}(x)$ and $p_{j,N/2}(x)$ (respectively, $q_{j,N}(x)$ and $q_{j,N/2}(x)$) even (respectively, odd) polynomials, $p_{e,i}(x)$ a polynomial, $e$, $e_j^+$ and $e_j^-$ distinct complex numbers. Also, $\cosh_q(x)=\frac{q^x+q^{-x}}{2}$ and $\sinh_q=\frac{q^x-q^{-x}}{2}$. The last two expresions in \eqref{eq:formulaspn} are unique up to a sign of $e^+_j$ or a simultaneous change of signs of $e^-_j$ and the respective $q_j(x)$. We call  $e^+_j$ (respectively,  $e^-_j$), \textit{even type} (respectively \textit{odd type}) \textit{exponents of} $V$ with \textit{multiplicities} $p_j(x)$ (respectively, $q_j(x)$). As in \cite{KWY}, we denote $e^+$ the set of even type exponents $e^+_j$ with multiplicity  $p_j(x)$ and by $e^-$ the set of odd type exponents $e^-_j$ with multiplicity $q_j(x)$. Therefore, the pair $(e^+;e^-)$ uniquely determines $V$. Analogously for the first formula, we call ${e_i}$ the \textit{exponents of} $V$ with \textit{multiplicities} $p_{e,i}(x)$, and we denote $e$ the set of exponents $e_i$ with multiplicity $p_{e,i}(x)$. We will denote this module by $L(\sqnh; e; e^+; e^-)$.


\section{Interplay between $\sqnh$ and the infinite rank classical Lie algebras}
In this section we will discuss the connection between $\sqnh$ and the Lie algebra of infinite matrices with finitely many nonzero diagonals over the algebra of truncated polynomials and its classical Lie subalgebras.
Let $\O$ be the algebra of all holomorphic functions on $\CC^{\times}$ with the topology of uniform convergence on compact sets, and denote

\begin{equation*}
\O^{\epsilon,j}=\{ f \in \O | f(w)=-\epsilon^j f(w^{-1}) \}.
\end{equation*}

Let $R$ be an associative algebra over $\CC$ and denote $\ri$ a free $R$-module with a fixed basis $\{ v_j \}_{j \in \ZZ}$ and denote $R_m = \CC [t]/(t^{m+1})$ where $m \in \ZZ_+$.

We consider the vector space $\sqnoa$ spanned by the quantum pseudo differential operators (of infinite order) of the form $z^kf(\tq)E_{i,j}$, where $f \in \O$. The bracket in $\sqn$ extends to $(\sqn)^{\O}$. In a similar fashion, we define a completion $(\sqns)^{O}$ of $\sqns$ consisting of all pseudo differential operators of the form

\begin{align*}
\{ z^k(q^{(k-1)/2}T_q)^{rk/2}(f(q^{(k-1)/2}T_q)E_{i,j}-(\epsilon)^kf((q^{(k-1)/2}T_q)^{-1})E_{N+1-j,N+1-i}):\\
 k \in \ZZ, \, 1 \leq i <j \leq N, \, f \in \O \},
\end{align*}
and the opposite diagonal
\begin{equation*}
\{ z^k(q^{(k-1)/2}T_q)^{rk/2}(f(q^{(k-1)/2}T_q)E_{i,N+1-i}: k \in \ZZ, \, 1 \leq i  \leq N, \, f \in \O^{\epsilon,k} \}.
\end{equation*}

Then the $2$-cocycle $\psi$ on $\sqns$ extends to a $2$-cocycle $\psi$ on $(\sqns)^{\O}$. Recall that $\sqn^{\prime}$ denotes the derived algebra of $\sqn$. Let $\sqnoh = \sqn^{\O \prime} +\CC C$ be the corresponding central extension.

Given $s \in \CC$, we have (cf. (3.2) in \cite{BL02}) the embedding $\varphi^{[m]}_s:\sqn \longrightarrow \gl $ ($\varphi^{[m]}_s:(\sqn)^{\O} \longrightarrow \gl $) given by

\begin{equation*}
\varphi^{[m]}_{s}(z^kf(\tq)E_{i,j})=\sum_{l \in \ZZ}f(sq^{-l+t})E_{(l-k)N-i+1, lN-j+1}
\end{equation*}
which are Lie algebra homomorphisms. A restriction of these homomorphisms of Lie algebras to $\sqns$ gives a family of homomorphisms of Lie algebras $\varphi^{[m]}_s:\sqns \longrightarrow \gl $ ($\varphi^{[m]}_s:(\sqns)^{\O} \longrightarrow \gl $).


For each $s \in \CC$ and $k \in \ZZ$, set

\begin{equation*}
I^{[m]}_{s,k}=\{ f \in \O: f^{(i)}(sq^{(k-1)/2 +n})=0 \, \text{ and } \, f^{(i)}(s^{-1}q^{-(k-1)/2 -n})=0,  \, \forall \, n \in \ZZ, 0 \leq i \leq m \}
\end{equation*}
and
\begin{equation*}
\tilde{I}^{[m]}_{s,k,\epsilon}=\{ f \in \O^{\epsilon,j}: f^{(i)}(sq^{(k-1)/2 +n})=0  \quad \forall \, n \in \ZZ, 0 \leq i \leq m \}.
\end{equation*}

Let
\begin{align*}
J^{[m],r,\epsilon}_{s}=\oplus_{k \in \ZZ} \{ z^k(q^{(k-1)/2}T_q)^{rk/2}(f(q^{(k-1)/2}T_q)E_{i,j}-(\epsilon)^kf((q^{(k-1)/2}T_q)^{-1})E_{N+1-j,N+1-i}):\\
 \, 1 \leq i <j \leq N, \, f \in I^{[m]}_{s,k} \} \\
\oplus \, \oplus_{k \in \ZZ} \{ z^k(q^{(k-1)/2}T_q)^{rk/2}(f(q^{(k-1)/2}T_q)E_{i,N+1-i}:  \, 1 \leq i  \leq N, \, f \in \tilde{I}^{[m]}_{s,k,\epsilon} \}.
\end{align*}

Using the Taylor formula on $\varphi^{[m]}_s: \sqns \longrightarrow \gl$, it follows that

\begin{equation}\label{eq:kernel}
\ker \varphi^{[m]}_s = J^{[m],r,\epsilon}_{s}.
\end{equation}

Choose a branch of $\log q$. Let $\tau = \log q/ 2\pi i$. Then any $s \in \CC^{\times}$ is uniquely written as $s=q^a$, with $a \in \CC/\tau^{-1}\ZZ$. Fix $\vec{s}=(s_1, \cdots, s_M) \in \CC^M$ such that if each $s_i=q^{a_i}$, we have

\begin{equation}\label{eq:aiaj}
a_i-a_j \notin \ZZ + \tau^{-1} \ZZ \quad \text{ for } i \neq j,
\end{equation}
and $\vec{m}=(m_1, \cdots, m_M) \in \ZZ^{M}_{\geq 0}$.
Let ${ g\ell}_\infty^{[\vec{m}]}=\oplus^M_{i=1} { g\ell}_\infty^{[m_i]}$. Consider the homomorphism

\begin{equation*}
\varphi^{[\vec{m}]}_s = \oplus^M_{i=1} \varphi^{[m_i]}_{s_i}: \sqnso \longrightarrow { g\ell}_\infty^{[\vec{m}]}.
\end{equation*}

\begin{proposition}\label{prop:aenz2}
Given $\vec{s}$ and $\vec{m}$ as above, we have the exact sequence of $\ZZ-$graded Lie algebras, provided that $|q| \neq 1:$
\begin{equation}
0 \rightarrow J^{[\vec{m}], r, \epsilon}_{\mathbf{s}} \rightarrow \sqnso \rightarrow { g\ell}_\infty^{[\vec{m}]} \rightarrow 0,
\end{equation}
where $J^{[\vec{m}], r, \epsilon}_{\mathbf{s}}= \cap^M_{i=1} J^{{[m_i], r, \epsilon}}_{s_i}$.
\end{proposition}

\begin{proof}
The injectivity part is clear from \eqref{eq:kernel}. For the sake of simplicity, we will prove the surjectivity of $\varphi^{[\vec{m}]}_s$ for the case $M=1$, $\vec{m}=m$ and $\vec{s}=s=q^a$. We will make use of the well-known fact that for every discrete sequence of points of $\CC$ and a non-negative integer $m$ there exists $f(w) \in \O$ having prescribed values of its first $m$ derivatives at these points. By conditions \eqref{eq:aiaj} and $|q| \neq 1$ and since $a \notin \ZZ/2$ we have that $\{ q^{(n-1)/2+j+a} \}$ and $\{ q^{-(n-1)/2-j-a} \}$ are discrete and disjoint sequences of points in $\CC$. Therefore we can find $f \in \O$ such that every element $t^j E_{a,b}$ is in the image, finishing the proof.
\end{proof}

We now intend to extend the homomorphism $\fims$ to a homomorphism between the central extensions of the corresponding Lie algebras.

\begin{proposition}\label{prop:fihat}
The $\CC-$linear map $\hat{\varphi}^{[m]}_s: \sqnh \rightarrow \glh$ defined by ($s=q^a$),
\begin{equation}
\hat{\varphi}^{[m]}_s|_{(\sqnh)_j}=\fims|_{(\sqns)_j} \quad \text{ if } j \neq 0,
\end{equation}
\begin{align}
\hat{\varphi}^{[m]}_s(q^{-n/2}\tq^nE_{i,i}-q^{n/2}\tq^{-n}E_{N+1-i,N+1-i})= \fims(q^{-n/2}\tq^nE_{i,i}-q^{n/2}\tq^{-n}E_{N+1-i,N+1-i}) \\ \nonumber
- \sum^m_{j=1} \frac{q^{(a-1)n}+(-1)^jq^{(-a+1)n}}{q^{n/2}-q^{-n/2}}(n \log q)^j \frac{t^j}{j!} \quad (n \neq 0),
\end{align}
\begin{equation}
\hat{\varphi}^{[m]}_s(C)=1 \in R_m,
\end{equation}
is a Lie algebra homomorphism over $\CC$.
\end{proposition}

\begin{proof}
It is a straightforward computation restricting the formula $\hat{\varphi}^{[m]}_s$ in (3.2) of \cite{BL02}, to $\sqnh$.
\end{proof}

The homomorphism $\fims$ is defined for any $s \in \CC$. However, for $a \in \ZZ/2$, it is no longer surjective. These cases are described by the following propositions.

\begin{proposition}\label{prop:a1}
For $a=1$, we have the following exact sequence of Lie algebras:
\begin{align*}
0 \rightarrow J^{[m],k,\epsilon}_{s} \rightarrow (\sqns)^{\O} \rightarrow \mathfrak{g} \rightarrow 0
\end{align*}
where $\mathfrak{g} \simeq \bar{d}^{[m]}_{\infty}$ if $\epsilon=1$ and $\mathfrak{g} \simeq \bar{c}^{[m]}_{\infty}$ if $\epsilon=-1$.
\end{proposition}

\begin{proof}
We will first prove the case $\epsilon=1$. The homomorphism $\fims:\sqn \rightarrow \gl$ introduced in \cite{BL02} is surjective. The anti-involution of $\sqn$ defined in \eqref{eq:sigma} transfers, via $\fims$, to an anti-involution $\omega: \gl \rightarrow \gl$ as follows
\begin{equation}
\omega(u^kE_{i,j})=(-u)^kE_{1-j,1-i}.
\end{equation}
Therefore, the Lie algebra of $-\sigma$ fixed points in $\sqn$, explicitly, $\sqns$, maps surjectively to the Lie algebra of $-\omega$ fixed points in $\gl$, explicitly, $\bar{d}^{[m]}_{\infty}$.
If $\epsilon=-1$, the anti-involution $\omega$ is as follows
\begin{equation*}
\omega(u^kE_{i,j})=(-1)^{q_j-q_i}(-u)^kE_{1-j,1-i},
\end{equation*}
where $i=q_iN+r_i$ and $j=q_jN+r_j$, with $1 \leq r_i \leq N$ and $1 \leq r_j \leq N$. As a result of the surjectivity described, it is enough to show that $\omega$ is conjugated by an automorphism $T^{\prime}$ of $\gl$ to the anti-involution defining $\bar{c}^{[m]}_{\infty}$. To that end, we define
\begin{equation}\label{eq:ta1}
T^{\prime}(u^mE_{a,b}-\epsilon^{q_b-q_a}(-u)^mE_{1-b,1-a} )=u^mE_{a,b}-\epsilon^{a+b}(-u)^mE_{1-b,1-a},
\end{equation}
where $a=q_aN+r_a$ and $b=q_bN+r_b$, with $0 \leq r_a \leq N-1$ y $0 \leq r_b \leq N-1$.
It is easy to see that $\omega$ is conjugated by $T^{\prime}$ to the anti-involution defining $\bar{c}^{[m]}_{\infty}$.
\end{proof}

\begin{proposition}\label{prop:a1/2nimpar}
If $a=1/2$ and $N$ is odd, we have the following exact sequence of Lie algebras:
\begin{align*}
0 \rightarrow J^{[m],k,\epsilon}_{s} \rightarrow (\sqns)^{\O} \rightarrow \mathfrak{g} \rightarrow 0,
\end{align*}
where $\mathfrak{g} \simeq \bar{b}^{+[m]}_{\infty}$ if $\epsilon=1$ and $\mathfrak{g} \simeq \bar{b}^{-[m]}_{\infty}$ if $\epsilon=-1$.
\end{proposition}

\begin{proof}
If $\epsilon = 1$, replace in the proof of the last proposition $\omega$ by
\begin{equation}
\omega(u^kE_{i,j})=(-u)^kE_{-N+1-j,-N+1-i}.
\end{equation}
Therefore, the Lie algebra of $-\sigma$ fixed points in $\sqn$, explicitly, $\sqns$, maps surjectively to the Lie algebra of $-\omega$ fixed points in $\gl$. Consequently, it is enough to see that $\omega$ is conjugated by an automorphism $T$ of $\gl$ to the anti-involution defining $\bar{b}^{+[m]}_{\infty}$. So, we define
\begin{equation}\label{eq:timpar}
T(u^r E_{i,j})=u^r E_{(-N+1)/2+i,(-N+1)/2+j}.
\end{equation}
It is easy to check that this extends to an automorphism of the algebra $\gl$ that conjugates $\omega$ to the anti-involution defining $\bar{b}^{+[m]}_{\infty}$.
If $\epsilon = -1$, $\omega$ is the following
\begin{equation}
\omega(u^kE_{i,j})=(-1)^{q_i-q_j}(-u)^kE_{-N+1-j,-N+1-i},
\end{equation}
where $i=q_iN+r_i$ and $j=q_jN+r_j$, with $1 \leq r_i \leq N$ y $1 \leq r_j \leq N$. The automorphism of $\gl$ for this case is $D=T \circ T^{\prime}$, where $T$ is the same that in the previous case and we have
\begin{equation}\label{eq:timparmenosuno}
T^{\prime}(u^mE_{a,b}-\epsilon^{q_b-q_a}(-u)^mE_{-b,-a} )=u^mE_{a,b}-\epsilon^{a+b}(-u)^mE_{-b,-a},
\end{equation}
with $a=q_aN+r_a$ and $b=q_bN+r_b$, for $0 \leq r_a \leq N-1$ and $0 \leq r_b \leq N-1$.
It is easy to see that $\omega$ is conjugated by $D$ to the anti-involution defining $\bar{b}^{-[m]}_{\infty}$.
\end{proof}

\begin{proposition}\label{prop:a1/2npar}
If $a=1/2$ and $N$ is even, we have the following exact sequence of Lie algebras:
\begin{align*}
0 \rightarrow J^{[m],k,\epsilon}_{s} \rightarrow (\sqns)^{\O} \rightarrow \bar{d}^{[m]}_{\infty} \rightarrow 0.
\end{align*}
\end{proposition}

\begin{proof}
This proof follows the same steps as last proposition. If $\epsilon=1$, because $\omega$ is the same as before, it is enough to replace $T$ by
\begin{equation}\label{eq:tpar}
T(u^r E_{i,j})=u^r E_{-N/2+i,-N/2+j}.
\end{equation}
The rest of the proof is the same for this case. If $\epsilon=-1$, $\omega$ is the same formula as in the last proposition, so it is enough to replace $T^{\prime}$ by
\begin{equation}\label{eq:tparmenosuno}
T^{\prime}(u^mE_{a,b}-\epsilon^{q_b-q_a}(-u)^mE_{1-b,1-a} )=u^mE_{a,b}-(-u)^mE_{1-b,1-a},,
\end{equation}
where $a=q_aN+r_a$ and $b=q_bN+r_b$, with $0 \leq r_a \leq N-1$ and $0 \leq r_b \leq N-1$.
\end{proof}

\begin{remark}\label{rmk:embeding}
\begin{enumerate}
\item[(a)] By an abuse of notation, for $a=1$ and $a=1/2$, in view of Propositions \ref{prop:a1} to \ref{prop:a1/2npar}, we will denote again $\fims$ the surjective homomorphism from $\sqns$ onto $\bar{c}^{[m]}_{\infty}$, $\bar{b}^{[m]}_{\infty}$ and $\bar{d}^{[m]}_{\infty}$, respectively, given by the old $\fims$ composed with the corresponding isomorphisms introduced in the proof of the proposition above.
\item[(b)] Recall that $\nu$ was defined in \eqref{eq:nu}. If $\epsilon=1$, for arbitrary $a \in \ZZ$, the image of $\sqns$ under the homomorphism $\varphi^{[m]}_{q^a}$ is $\nu^{a}(\di)$. Similarly, if $a \in \ZZ+1/2$, the image of $\sqns$ under the homomorphism $\varphi^{[m]}_{q^a}$ is $\nu^{a}(\di)$ if $N$ is even and $\nu^{a}(\bi)$ if $N$ is odd. As a consequence, it is enough to study the cases $a=1$ and $a=1/2$. The same conclusions can be obtained for $\epsilon=-1$. Therefore, we will only consider  $a=1$ and $a=1/2$ throughout this paper.
\end{enumerate}
\end{remark}

Given vectors $\vec{s}=(s_1, \cdots, s_M)=(q^{a_1}, \cdots, q^{a_M}) \in \CC^M$ and $\vec{m}=(m_1, \cdots, m_M) \in \ZZ^M$ such that if $a_i \in \ZZ$, then $a_i=1$; if $a_i \in \ZZ+1/2$ then $a_i=1/2$; and $a_i-a_j \notin \ZZ+\tau^{-1} \ZZ$ for $i \neq j$. Combining this with Propositions \ref{prop:aenz2} to \ref{prop:a1/2npar}, we obtain a surjective Lie algebra homomorphism
\begin{equation}\label{eq:extensionfi}
\varphi^{[\vec{m}]}_{\vec{s}} = \oplus^n_{i=1} \varphi^{[m_i]}_{s_i}: \widehat{\sqnso} \longrightarrow \mathfrak{g}^{[\vec{m}]}:= \sum^n_{i=1}\mathfrak{g}^{[m_i]},
\end{equation}
where if $\epsilon=1$
\begin{align*}
\mathfrak{g}^{[m_i]}= \begin{cases}
\widehat{ g\ell}_\infty^{[m_i]} & \quad \text{ if } a_i \notin \ZZ/2,  \\
\tilde{b}^{[m_i]}_{\infty} & \quad \text{ if } a_i = 1/2 \, \text{ and N is odd}, \\
d^{[m_i]}_{\infty} & \quad \text{ if } a_i = 1/2 \, \text{ and N is even or } \, a_i=1.\\
\end{cases}
\end{align*}
and if $\epsilon=-1$
\begin{align*}
\mathfrak{g}^{[m_i]}= \begin{cases}
\widehat{ g\ell}_\infty^{[m_i]} & \quad \text{ if } a_i \notin \ZZ/2,  \\
b^{[m_i]}_{\infty} & \quad \text{ if } a_i = 1/2 \, \text{ and N is odd}, \\
d^{[m_i]}_{\infty} & \quad \text{ if } a_i = 1/2 \, \text{ and N is even}, \\
\ci & \quad \text{ if } a_i=1.
\end{cases}
\end{align*}

\section{Realization of quasifinite highest weight modules of $\sqnh$}

In this section $\gm$ will be $\glh$ or one of its classical subalgebras. The proof of the following Proposition is standard (cf. \cite{K})

\begin{proposition}
The $\gm-$module $L(\gm, \lambda)$ is quasifinite if and only if all but finitely many of the $^{\dagger}h^{(i)}_j$ are zero, where $\dagger$ represents $a, \, b, \, c$ or $d$ depending on whether $\gm$ is $\glh,\, \bi, \, \ci$ or $\di$.
\end{proposition}

Given $\vec{m}=(m_1, \cdots, m_M) \in \ZZ^M_{\geq 0}$, take a quasifinite $\lambda_i \in (\mathfrak{g}^{[m_i]})^\ast_0$ for each $1 \leq i \leq M$, and let $L(\mathfrak{g}^{[m_i]}, \lambda_i)$ be the corresponding $\mathfrak{g}^{[m_i]}$-module. Let $\vec{\lambda}=(\lambda_1, \cdots, \lambda_M)$. Then the tensor product
\begin{equation}
L(\gm, \lambda)=\otimes^M_{i=} L(\mathfrak{g}^{[m_i]}, \lambda_i)
\end{equation}
is an irreducible $\mathfrak{g}^{[\vec{m}]}$-module, with $\mathfrak{g}^{[\vec{m}]}=\oplus^M_{i=1}\mathfrak{g}^{[m_i]}$. The module $L(\mathfrak{g}^{[\vec{m}]}, \vec{\lambda})$ can be regarded as a $\sqnh-$-module via the homomorphism $\varphi^{[\vec{m}]}_{\vec{s}}$ and will be denoted by $L^{[\vec{m}]}_{\vec{s}}(\vec{\lambda})$. We shall need the following results.

\begin{proposition}\label{prop:extensionmod}
Let $V$ be a quasifinite $\sqnh$-module. Then the action of $\sqnh$ on $V$ naturally extends to the action of $(\sqnh)^{\O}_u$ on $V$, for any $u \neq 0$.
\end{proposition}

\begin{proof}
The proof is similar to the proof of Proposition (4.3) of \cite{KL}, replacing $B=adD^2-k^2$ by the following:
\begin{itemize}
\item If $i \neq j$, $i \neq N+1-j$, $i \neq N+1-i$ and $j \neq N+1-j$ \\
\begin{align*}
B=& \frac{1}{2q^k}(\text{ad} (\tq) E_{i,i}-\text{ad} (q^k\tq) E_{j,j})\\
&+\frac{1}{2}(\text{ad} (q\tq^{-1}) E_{N+1-j,N+1-j}-\text{ad} (q^{-k+1}\tq) E_{N+1-i,N+1-i}) .
\end{align*}

\item If $i = j$,
\begin{align*}
B=\frac{1}{q^k-1}\text{ad} \tq E_{i,i}+\frac{q^{-k+1}}{q^{-k}-1}\text{ad} \tq^{-1} E_{N+1-i,N+1-i}.
\end{align*}

\item If $i = N+1-j$,
\begin{align*}
B=\frac{1}{q^k}\text{ad} \tq E_{i,i}-q^{-k+1}\text{ad} \tq^{-1} E_{N+1-i,N+1-i}.
\end{align*}
\end{itemize}
\end{proof}

\begin{theorem}\label{teo:lirreduc}
Let $V$ be a quasifinite $\mathfrak{g}^{[\vec{m}]}$-module, which is regarded as a $\sqnh$-module via the homomorphism $\varphi^{[\vec{m}]}_{\vec{s}}$. Then any $\sqnh$-submodule of $V$ is also a $\mathfrak{g}^{[\vec{m}]}$-submodule. In particular, the $\sqnh$-module $L^{[\vec{m}],k,\epsilon}_{\vec{s}}(\vec{\lambda})$ is irreducible if $\vec{s}=(s_1, \cdots, s_M)$ with $s_i=q^{a_i}$ is such that $a_i \in \ZZ$ implies $a_i=1$; $a_i \in \ZZ+1/2$ implies $a_i=1/2$ and $a_i-a_j \notin \ZZ+\tau^{-1} \ZZ$ for $i \neq j$.
\end{theorem}

\begin{proof}
Let $W$ be a $\sqnh$-submodule of $V$. Due to the fact that $W$ is a quasifinite $\sqnh$-module as well, by Proposition \ref{prop:extensionmod} it can be extended to $(\sqnh)^{\O}_u$ for $u \neq 0$. As a result of \eqref{eq:extensionfi}, the map $\varphi^{[\vec{m}]}_{\vec{s}}:(\sqnh)^{\O}_u \longrightarrow (\mathfrak{g}^{[\vec{m}]})_u$ is surjective for any $u \neq 0$. Therefore, $W$ is invariant with respect to all members of the principal gradation of $(\mathfrak{g}^{[\vec{m}]})_u$ with $u \neq 0$. Since $\mathfrak{g}^{[\vec{m}]}$ coincides with its derived algebra, this proves the theorem.
\end{proof}


Now, we will proceed to show that all the irreducible quasifinite $\sqnh$-modules can be realized as some $L^{[\vec{m}],k,\epsilon}_{\vec{s}}(\vec{\lambda})$, for some $\vec{m} \in \ZZ^M_{\geq 0}$ and $\vec{s} \in \CC^M$, with $s_i=q^{a_i}$ such that $a_i-a_j \notin \ZZ+\tau^{-1} \ZZ$ for $i \neq j$. For simplicity, we will consider the case $M=1$ to calculate the generating series $\bigtriangleup^{ \epsilon}_{m,s,\lambda,i}(x)=\sum_{n \in \ZZ}(\bigtriangleup^{\epsilon}_{m,s,\lambda,i})_n \, x^{-n} $ of the highest weight and central charge $c$ of the $\sqnh$-module $L^{[m],k,\epsilon}_s(\lambda)$.

$ $\newline
We will introduce the following notation
\begin{equation}
\eta_i(\alpha,\beta)=\frac{q^{\alpha\beta}+(-1)^iq^{-\alpha\beta}}{q^{\beta/2}-q^{-\beta/2}} \frac{(\beta \log q)^i}{i!}.
\end{equation}

Making use of Theorem \eqref{eq:caracterizacion}, take an irreducible quasifinite weight $\sqnh$-module $V$ with central charge $c$ and generating series $\bigtriangleup_i(x)$, $P^{\epsilon}_N(x)$ an even quasipolynomial such that
\begin{align}
P^{\epsilon}_N(n)=\bigtriangleup_{N,n} \quad \text{ for } \quad  n \neq 0 \quad \text{ and } \quad P^{\epsilon}_N(0)=-2c,
\end{align}
$P_{i}(x)$ a quasipolynomial such that
\begin{equation}
P_{i}(n)=\bigtriangleup_{i,n}-\bigtriangleup_{i+1,n}
\end{equation}
for $1 < i \leq [N/2]-\delta_{N, even}$ are quasipolynomials, and when $N$ is even, $P_{N/2}(x)$ an even quasipolynomial such that
\begin{equation}
P_{N/2}(n)=\bigtriangleup_{N/2,n}-\bigtriangleup_{N/2+1,n}.
\end{equation}
We write
\begin{align}\label{eq:cuaspipolpn}
&P_{i}(x)=\sum_{s \in \CC} p_{s,i}(x)q^{s_ix}, \quad \text{ for } \quad 1 < i \leq [N/2]-\delta_{N,even} \\ \nonumber
&P^{\epsilon}_N(x)=\sum_{j \in \ZZ}p^{\epsilon}_{j,N}(x)\cosh_q(e_j^+x)+\sum_{j \in \ZZ}q^{\epsilon}_{j,N}(x)\sinh_q(e_j^-x) \quad \text{ and } \\ \nonumber
&P_{N/2}(x)=\sum_{j \in \ZZ}p_{j,N/2}(x)\cosh_q(e_j^+x)+\sum_{j \in \ZZ}q_{j,N/2}(x)\sinh_q(e_j^-x),
\end{align}
where $p_{j,N}(x)$ and $p_{j,N/2}(x)$ (respectively, $q_{j,N}(x)$ and $q_{j,N/2}(x)$) are even (respectively, odd) polynomials and $p_{e,i}(x)$ is a polynomial. Let $L^{[\vec{m}]}_{\vec{s}}(\mathfrak{g}^{[\vec{m}]}, \vec{\lambda})$ be a representation of $\mathfrak{g}^{[\vec{m}]}$ considered as a representation of $\sqnh$ via $\hat{\varphi}^{[\vec{m}]}_s$, where $\mathfrak{g}^{[m]}$ is ${ g\ell}_\infty^{[m]}$ or one of its classical subalgebras. Then
\begin{equation}\label{eq:deltafii}
(\bigtriangleup^{\epsilon}_{m,a,\lambda,i})_n=-\lambda(\hat{\varphi}^{[m]}_s((q^{-1/2}\tq)^nE_{i,i}-(q^{1/2}\tq)^{-n}E_{N+1-i,N+1-i}),
\end{equation}
with $1 < i \leq [N/2]+\delta_{N,even}$, and
\begin{equation}\label{eq:deltafiN}
(\bigtriangleup^{\epsilon}_{m,a,\lambda,N})_n=-\lambda(\hat{\varphi}^{[m]}_s((\tq^n+\tq^{-n})E_{N,N}-((q^{-1}\tq)^n+(q^{-1}\tq)^{-n})E_{1,1}),
\end{equation}
where $\hat{\varphi}^{[m]}_s$ is the embedding given by proposition \ref{prop:fihat} composed accordingly with the isomorphisms defined in propositions \ref{prop:a1} to \ref{prop:a1/2npar}.



\begin{proposition}\label{prop:casoanoenz/2}

Take the embedding $\hat{\varphi}^{[m]}_s: {\sqnh} \longrightarrow \glh$ with $s=q^a$ y $a \notin \ZZ/2$. The $\glh$-module $L(\glh,\lambda)$ regarded as a ${\sqnh}$-module is isomorphic to \\
$L({\sqnh};e; e^+;e^-)$, where
\begin{enumerate}
\item[(a)] The exponents $e$ are $-1/2+a-l$ and $1/2-a+l$ with $l \enz$ and their respective mutiplicities are
\begin{align}
& p_{1/2-a+l,i}(x)=\sum^m_{u=0} \frac{(x\log q)^u}{u!} {^ah^{(u)}_{(l-1)N+i}} \quad \text{ and } \\ \nonumber
& p_{-1/2+a-l,i}(x)=\sum^m_{u=0} \frac{(-x\log q)^u}{u!} {^ah^{(u)}_{lN-i}},
\end{align}
for $1 < i \leq [N/2]-\delta_{N,even}$.
\item[(b)] The exponents are $e^+=e^-=a-l$ with $l \enz$ with multiplicities
\begin{equation}
p^{\epsilon}_{a-l,N}(x)=\sum^m_{u=0, u \, even}{^a\widehat{h}^{(u)}_{(l-1)N}}\frac{x^u}{u!} \quad \text{ and } \quad q^\epsilon_{a-l,N}(x)=\sum^m_{u=0, u \, odd}{^a\widetilde{h}^{(u)}_{(l-1)N}}\frac{x^u}{u!},
\end{equation}
where
\begin{align*}
&{^a\widehat{h}^{(u)}_{(l-1)N}}=2(\log q)^u({^ah^{(u)}_{(l-1)N}}+\delta_{l,1} (c_u-\delta_{u,0}c_0)) \quad \text{ and } \\
&{^a\widetilde{h}^{(u)}_{(l-1)N}}= 2(\log q)^u ({^ah^{(u)}_{(l-1)N}}+\delta_{l,1}c_u), \nonumber
\end{align*}
\end{enumerate}
and $P^{\epsilon}_{l,N}(0)=-2c_0$ for $i=N$.
\item[(c)] Moreover, if $N$ is even, we have the exponents $e^+=e^-=1/2-a+l$ with $l \enz$ and their multiplicities are
\begin{align}
&p_{1/2-a+l,N/2}(x)=\sum^m_{u=0, u \text{ par}} \frac{(x \log q)^u}{u!}{^ah^{(u)}_{(l-1/2)N}} \quad \text{ and } \\ \nonumber
&q_{1/2-a+l,N/2}(x)=\sum^m_{u=0, u \text{ impar}} \frac{-(x \log q)^u}{u!}{^ah^{(u)}_{(l-1/2)N}},
\end{align}
for $i=N/2$.
\end{proposition}

\begin{proof}
If $1 < i \leq [N/2]-\delta_{N,even}$, combining the formulas of proposition \ref{prop:fihat} with \eqref{eq:gl} and \eqref{eq:deltafii}, we have that
\begin{align*}
(\bigtriangleup^{\epsilon}_{m,a,\lambda,i})_n =  \sum_{l \in \ZZ} \sum^m_{u=0} \bigg( & \frac{-(n\log q)^u}{u!}q^{n(-1/2+a-l)} \,\, {^a\lambda^{(u)}_{lN+1-i}}  \\ \nonumber
&+ \frac{(-n\log q)^u}{u!}q^{-n(-1/2+a-l)}\,\, {^a\lambda^{(u)}_{(l-1)N+i}}       \bigg)
 +\sum^m_{u=1}\eta_u(a-1,n)c_u.
\end{align*}
Then,
\begin{align*}
(\bigtriangleup^{\epsilon}_{m,a,\lambda,i})_n-(\bigtriangleup^{\epsilon}_{m,a,\lambda,i+1})_n=\sum_{l \in \ZZ}\sum^m_{u=0} &\frac{(-n\log q)^u}{u!} {^ah^{(u)}_{(l-1)N+i}q^{-n(-1/2+a-l)}}\\
&+\frac{(n\log q)^u}{u!}{^ah^{(u)}_{lN-i}q^{n(-1/2+a-l)}}.
\end{align*}
Making use of the definitions of multiplicities and exponents for the quasipolynomial $P_i(x)$ in \eqref{eq:cuaspipolpn}, we complete the proof for (a).

If $i=N$, as before, considering \eqref{eq:deltafiN} and \eqref{eq:gl}, we obtain
\begin{align*}
(\bigtriangleup^{\epsilon}_{m,a,\lambda,N})_n=& \sum_{l \in \ZZ} \sum^m_{u=0}  \bigg( \eta_u(-1+a-l,n) 2 \sinh_q(n/2) \,\, {^a\lambda^{(u)}_{lN}} \\
&-\eta_u(a-l,n)2 \sinh_q(n/2)\,\,{^a\lambda^{(u)}_{(l-1)N+1}} \bigg) +\sum^m_{u=1}2 \sinh_q(n/2)\eta_u(a-1,n)c_u.
\end{align*}
Shifting the index $l$ to $l-1$ in the first sum, we get
\begin{equation*}
(\bigtriangleup^{\epsilon}_{m,a,\lambda,N})_n= \sum_{l \in \ZZ} \sum^m_{u=0} 2 \sinh_q(n/2) \eta_u(a-l,n)({^ah^{(u)}_{(l-1)N}}+\delta_{l,1}c_r-c_0).
\end{equation*}
Since
\begin{equation*}
  2 \sinh_q(n/2) \eta_u(a-l,n)=\frac{(n \log q)^u}{u!}(q^{n(a-l)}+(-1)^uq^{-n(a-l)})
\end{equation*}
and making use of the definitions of multiplicities and exponents for the quasipolynomials $P^\epsilon_N(x)$ in \eqref{eq:cuaspipolpn}, we finish the proof for (b).

If $N$ even, following the same steps as in the proof of (a), we have
\begin{equation*}
(\bigtriangleup^{\epsilon}_{m,a,\lambda,N/2})_n-(\bigtriangleup^{\epsilon}_{m,a,\lambda,N/2+1})_n=\sum_{l \in \ZZ} \sum^m_{u=0} \frac{(n\log q)^u}{u!}{^ah^{(u)}_{(l-1/2)N}}(q^{-n(1/2-a+l)}+(-1)^uq^{n(1/2-a+l)}).
\end{equation*}
Then, splitting the sums according to the parity of $u$, we get the multiplicities and exponents expected.

\end{proof}


\begin{proposition}

Let $s=q^a$ with $a=1/2$ and $N$ even. Take the embedding $\hat{\varphi}^{[m]}_s: {\sqnh} \longrightarrow \di$. The $\di$-module $L(\di,\lambda)$ regarded as a ${\sqnh}$-module is isomorphic to $L({\sqnh};e; e^+;e^-)$, where
\begin{enumerate}

\item[(a)] If $1 < i \leq [N/2]-\delta_{N,even}$, the exponents are $e=l$ with $l \enz$ and their multiplicities are
\begin{align}
& p_{l,i}(x)=\sum^m_{u=0} \frac{(-x\log q)^u}{u!} {^dh^{(u)}_{(l-1/2)N+i}} \quad \text{ if } \quad l>0 \quad \text{ and } \\ \nonumber
& p_{l,i}(x)=\sum^m_{u=0} \frac{(x\log q)^u}{u!} {^dh^{(u)}_{(l-1/2)N-i}} \quad \text{ if } \quad l\leq 0
\end{align}

\item[(b)] If $i=N$, the exponents $e^+$ and $e^-$ are $1/2-l$ with $l \geq 1$ and their multiplicities are
\begin{equation}
p^{\epsilon}_{1/2-l,N}(x)=\sum^m_{u=0, u \, even}{^d\widehat{h}^{(u)}_{(l-1/2)N}}\frac{x^u}{u!} \quad \text{ and } \quad q^\epsilon_{1/2-l,N}(x)=\sum^m_{u=0, u \, odd}{^d\widetilde{h}^{(u)}_{(l-1/2)N}}\frac{x^u}{u!},
\end{equation}
where
\begin{align*}
&{^d\widehat{h}^{(u)}_{(l-1/2)N}}=2 (\log q)^u({^dh^{(u)}_{(l-1/2)N}}+\delta_{l,1}(c_u-\delta_{u,0}c_0)) \quad \text{ and } \\
&{^d\widetilde{h}^{(u)}_{(l-1/2)N}}=2(\log q)^u({^dh^{(u)}_{(l-1/2)N}}+\delta_{l,1}c_u)  \nonumber
\end{align*}
and $P^{\epsilon}_{l,N}(0)=-2c_0$.

\item[(c)] Moreover, if $N$ is even, for $i=N/2$ the exponents $e^+$ and $e^-$ are $l\geq 0$ and their multiplicities are, if $l \geq 1$,
\begin{align}
& p_{l,N/2}(x)=\sum^m_{u=0, u \text{ even}} 2\frac{(x \log q)^u}{u!}{^dh^{(u)}_{lN}} \quad \text{ and } \\ \nonumber
& q_{l,N/2}(x)=\sum^m_{u=0, u \text{ odd}} -2\frac{(x \log q)^u}{u!}{^dh^{(u)}_{lN}}
\end{align}
and if $l=0$
\begin{align}
& p_{0,N/2}(x)=\sum^m_{u=0, u \text{ even}} 2\frac{(x \log q)^u}{u!}{^d\lambda^{(u)}_{1}} \quad \text{ and } \quad q_{0,N/2}(x)=0.
\end{align}
\end{enumerate}
\end{proposition}


\begin{proof}
Consider first the case $\epsilon=1.$ By Remark \ref{rmk:embeding}, part (a), we have that the embedding $\hat{\varphi}^{[m]}_s: \sqnh \longrightarrow \di$ is in fact the embedding given by proposition \ref{prop:fihat} composed by $T^{-1}$, where $T$ is the automorphism of $\gl$ defined in \eqref{eq:tpar}.

If $1 < i \leq [N/2]-\delta_{N,even}$, using \eqref{eq:deltafii} for the embedding in this case, we get
\begin{align}\label{eq:1}
&(\bigtriangleup^{\epsilon}_{m,1/2,\lambda,i})_n=\lambda \bigg( \sum_{l \in \ZZ} \sum^m_{u=0} -\frac{(n \log q)^u}{u!}q^{-nl}t^rE_{(l+1/2)N+1-i,(l+1/2)N+1-i} \\ \nonumber
&+\sum_{l \in \ZZ} \sum^m_{u=0}\frac{(-n \log q)^u}{u!}q^{nl}t^rE_{(l-1/2)N+i,(l-1/2)N+i} \bigg) + \sum^m_{u=1} \eta_u(-1/2,n)c_u. \nonumber
\end{align}

%

Making an adequate change of variables in $l$ and using \eqref{eq:di}, we get
\begin{equation*}
(\bigtriangleup^{\epsilon}_{m,1/2,\lambda,i})_n-(\bigtriangleup^{\epsilon}_{m,1/2,\lambda,i+1})_n=\sum_{l \geq 1} \sum^m_{u=0} \frac{(n\log q)^u}{u!} \big(
(-1)^uq^{nl} \, { ^dh_{(l-1/2)N+i}}+q^{n(-l+1)}\, {^dh_{(l-1/2)N-i}} \big).
\end{equation*}
Making use of the definitions of multiplicities and exponents for the quasipolynomial $P_i(x)$ in \eqref{eq:cuaspipolpn}, we finish the proof of (a).

Using \eqref{eq:deltafiN} for the embedding in this case, we get
\begin{align}\label{eq:2}
&(\bigtriangleup^{\epsilon}_{m,1/2,\lambda,N})_n= \\ \nonumber
&\lambda \bigg( \sum_{l \in \ZZ} \sum^{m}_{u=0}  \frac{(n \log q)^u}{u!} t^u \big( (-q^{n(1/2-l)}-(-1)^uq^{-n(1/2-l)})E_{(l-1/2)N+1,(l-1/2)N+1} \\ \nonumber
&+\sum_{l \in \ZZ} \sum^{m}_{u=0} \frac{(n \log q)^u}{u!} t^u (q^{-n(1/2+l)}+(-1)^uq^{n(1/2+l)})E_{(l+1/2)N,(l+1/2)N} \big) \bigg)  \\ \nonumber
&+\sum^m_{u=1}2 \sinh_q(n/2)\eta_u(-1/2,n)c_u. \nonumber
\end{align}
Once again, making an adequate change of variables in $l$, using \eqref{eq:di}, and taking into account the fact that
\begin{equation*}
2\eta_u(-1/2,n)\sinh_q(n/2)=\frac{(q^{n/2}+(-1)^uq^{-n/2})(n \log q)^u}{u!},
\end{equation*}
we have
\begin{align*}
(\bigtriangleup^{\epsilon}_{m,1/2,\lambda,N})_n= & \sum_{l \geq 1} {^dh^{(u)}_{(l-1/2)N}(n)}(q^{n(1/2-l)}+(-1)^uq^{-n(1/2-l)})\\
&+\sum^m_{u=0}\frac{(n \log q)^u}{u!}(q^{n/2}+(-1)^uq^{-n/2})(c_u-\delta_{u,0}c_0).
\end{align*}

In order to complete the proof, we study the parity of $u$ and split the sums accordingly. As a result of the definitions of multiplicities and exponents for the quasipolynomials $P^\epsilon_N(x)$ in \eqref{eq:cuaspipolpn}, we find the exponents and multiplicities expected for (b).

For $i=N/2$, we follow the same steps as in (a), but with an adequate change of variables in $l$. This way, we get
\begin{align*}
(\bigtriangleup^{\epsilon}_{m,1/2,\lambda,N/2})_n-(\bigtriangleup^{\epsilon}_{m,1/2,\lambda,N/2+1})_n=& \sum_{l \geq 1} \sum^m_{u=0} \frac{(n \log q)^u}{u!}(q^{-n(l-1)}+(-1)^uq^{n(l-1)}){ ^dh^{(u)}_{(l-1)N}} \\
& + \sum^m_{u=0} \frac{(n \log q)^u}{u!} (1+(-1)^u) \, {^d\lambda^{(u)}_1}.
\end{align*}
Lastly, by splitting this sums according to the parity of $u$, we find the exponents and multiplicities expected, finishing the proof for this case.

Consider now $\epsilon=-1.$ The embedding $\hat{\varphi}^{[m]}_s: \sqnh \longrightarrow \di$ is in this case the embedding given by proposition \ref{prop:fihat} composed by $D=T \circ T^{\prime}$, where $T^{\prime}$ is the automorphism of $\gl$ defined in \eqref{eq:tparmenosuno}. The results for this embedding are the same as for $\epsilon=1$.
\end{proof}


\begin{proposition}

Let $s=q^a$ with $a=1/2$ and $N$ odd and take the embedding $\hat{\varphi}^{[m]}_s:{\sqnh} \longrightarrow \gm$, where
$\gm= \tilde{b}^{[m]}_{\infty}$ if $\epsilon=1$ and $\gm=\bi$
if $\epsilon=-1$.
The $\gm$-module $L(\gm,\lambda)$ regarded as a ${\sqnh}$-module is isomorphic to $L({\sqnh};e; e^+;e^-)$, where
\begin{enumerate}

\item[(a)] The exponents $e$ are $l \in \ZZ$ and their multiplicities are
\begin{align}
& p_{l,i}(x)=\sum^m_{u=0} \frac{(-x\log q)^u}{u!} {^bh^{(u)}_{(l-1/2)N+i-1/2}} \quad \text{ if } \quad l> 0 \quad { and } \\
& p_{l,i}(x)=\sum^m_{u=0} \frac{(x\log q)^u}{u!} {^bh^{(u)}_{(l-1/2)N-i-1/2}} \quad \text{ if } \quad l \leq 0, \nonumber
\end{align}
for $1 < i \leq [N/2]-\delta_{N,even}$.
\item[(b)] The exponents are $e^+=e^-=1/2-l$ with $l \geq 0$ and their respective multiplicities are, for $l \leq 1$,
\begin{align}
& p^{\epsilon}_{1/2-l,N}(x)=\sum^m_{u=0, u \, even}2(\log q)^u \, {^bh^{(u)}_{(l-1/2)N-1/2}}\frac{x^u}{u!} \quad \text{ and } \\ \nonumber
& q^\epsilon_{1/2-l,N}(x)=\sum^m_{u=0, u \, odd}2(\log q)^u \,  {^bh^{(u)}_{(l-1/2)N-1/2}}\frac{x^u}{u!},
\end{align}
and for $l=0$,
\begin{align*}
& p^{\epsilon}_{1/2,N}(x)=\sum^m_{u=0, u \, even}2(\log q)^u\frac{x^u}{u!}(c_u-\delta_{u,0}c_0) \quad \text{ and }  \quad q^\epsilon_{1/2,N}(x)=\sum^m_{u=0, u \, odd}2(\log q)^u\frac{x^u}{u!}c_u,
\end{align*}
and $P^{\epsilon}_{l,N}(0)=-2c_0$ for $i=N$.
\end{enumerate}
\end{proposition}


\begin{proof}
Consider first $\epsilon=1$. By remark \ref{rmk:embeding}, part (a), we have that the embedding $\hat{\varphi}^{[m]}_s: \sqnh \longrightarrow \bi$ is in fact the embedding given by proposition \ref{prop:fihat} composed by $T^{-1}$, where $T$ is the automorphism of $\gl$ defined in \eqref{eq:timpar}.
 Using \eqref{eq:deltafii} for the embedding for this case, we have
\begin{align}\label{eq:3}
&(\bigtriangleup_{m,1/2,\lambda,i})_n=\lambda \bigg( \sum_{l \in \ZZ} \sum^m_{u=0}- \frac{(n \log q)^u}{u!}q^{-nl}t^rE_{(l-1/2)N-i+1/2,(l-1/2)N-i+1/2} \\ \nonumber
&+\sum_{l \in \ZZ} \sum^m_{u=0}\frac{(-n \log q)^u}{u!}q^{nl}t^rE_{((l-1/2)N+i-1/2,(l-1/2)N+i-1/2} \bigg) + \sum^m_{u=1} \eta_u(-1/2,n)c_u. \nonumber
\end{align}
Making an adequate change of variables in $l$ and using \eqref{eq:bi}, we have
\begin{align*}
(\bigtriangleup_{m,1/2,\lambda,i})_n = & \sum_{l \geq 1} \sum^m_{u=0} \frac{(n \log q)^u}{u!} ({ ^d\lambda^{(u)}_{(l-1/2)N+i-1/2}}(-1)^uq^{nl}-{^d\lambda^{(u)}_{(l-1/2)N-i+1/2}}q^{-n(l-1)}) \\ \nonumber
&+ \sum^m_{u=1} \eta_u(-1/2,n)c_u.
\end{align*}
Then,
\begin{align*}
&(\bigtriangleup_{m,1/2,\lambda,i})_n-(\bigtriangleup_{m,1/2,\lambda,i+1})_n=\\
&\sum_{l \geq 1} \sum^m_{u=0} \frac{(n \log q)^u}{u!} ((-1)^u q^{nl} \,
{ ^bh^{(u)}_{(l-1/2)N+i-1/2}}+q^{-n(l-1)}\, {^bh^{(u)}_{(l-1/2)N-i-1/2}}).
\end{align*}
Making use of the definitions of multiplicities and exponents for the quasipolynomial $P_i(x)$ in \eqref{eq:cuaspipolpn}, we finish the proof of (a).

Using \eqref{eq:deltafiN} for the embedding in this case, we have
\begin{align}\label{eq:4}
&(\bigtriangleup_{m,1/2,\lambda,N})_n=\\ \nonumber
& \lambda \bigg( \sum_{l \in \ZZ} \sum^{m}_{u=0}  \frac{(n \log q)^u}{u!} t^u \big( (-q^{n(1/2-l)}-(-1)^uq^{-n(1/2-l)})E_{(l-1/2)N+1/2,(l-1/2)N+1/2} \\ \nonumber
&+ \sum_{l \in \ZZ} \sum^{m}_{u=0} \frac{(n \log q)^u}{u!} t^u (q^{-n(1/2+l)}+(-1)^uq^{n(1/2+l)})E_{(l+1/2)N-1/2,(l+1/2)N-1/2} \big) \bigg)  \\ \nonumber
&+\sum^m_{u=1}2 \sinh_q(n/2)\eta_u(-1/2,n)c_u.  \nonumber
\end{align}
Once again, by making an adequate change of variables in $l$, using \eqref{eq:bi} and the fact that
\begin{equation*}
2\eta_u(-1/2,n)\sinh_q(n/2)=\frac{(q^{n/2}+(-1)^uq^{-n/2})(n \log q)^u}{u!},
\end{equation*}
we obtain
\begin{align*}
(\bigtriangleup_{m,1/2,\lambda,N})_n= &\sum_{l \geq 1} \sum^{m}_{u=0} \frac{(n \log q)^u}{u!}{^dh^{(u)}_{(l-1/2)N-1/2}}(q^{n(1/2-l)}+(-1)^{u}q^{-n(1/2-l)}) \\
&+\sum^m_{u=0}\frac{(n \log q)^u}{u!}(q^{n/2}+(-1)^uq^{-n/2})(c_u-\delta_{u,0}c_0).
\end{align*}
To complete the proof, we study the parity of $u$ and split the sums accordingly. As a result of the definitions of multiplicities and exponents for the quasipolynomials $P^\epsilon_N(x)$ in \eqref{eq:cuaspipolpn}, we find the exponents and multiplicities expected.

Consider now $\epsilon=-1.$ The embedding $\hat{\varphi}^{[m]}_s: \sqnh \longrightarrow \di$ is in this case the embedding given by proposition \ref{prop:fihat} composed by $D=T \circ T^{\prime}$, where $T^{\prime}$ is the automorphism of $\gl$ defined in \eqref{eq:timparmenosuno}. Proceeding in an analogous way as for the case $\epsilon=1$, we get the expected results.

\end{proof}


\begin{proposition}\label{prop:casoa1}

Let $s=q^a$ with $a=1$ and let the embedding $\hat{\varphi}^{[m]}_s: {\sqnh} \longrightarrow \gm$, where
$\gm=\di$  if $\epsilon=1$ and $\gm=\ci$ if $\epsilon=-1$.
The $\gm$-module $L(\gm,\lambda)$ regarded as a ${\sqnh}$-module is isomorphic to $L({\sqnh};e;e^+;e^-)$, where
\begin{enumerate}

\item[(a)] If $1 < i \leq [N/2]-\delta_{N,even}$, the exponents $e$ are $1/2-l$ with $l \enz$ and their multiplicities are
\begin{align}
& p_{1/2-l,i}(x)=\sum^m_{u=0} \frac{(x\log q)^u}{u!}  {^{\dagger}h^{(u)}_{lN-i}} \quad \text{ with } l >0 \quad \quad \text{ ahd } \\ \nonumber
& p_{1/2-l,i}(x)=\sum^m_{u=0} \frac{(-x\log q)^u}{u!} {^{\dagger}h^{(u)}_{(l-1)N+i}}, \quad \text{ with } l \leq 0 \quad
\end{align}
where $\dagger$ represents $c$ or $d$ depending on whether $\gm$ is $\ci$ or $\di$.

\item[(b)] If $i=N$, the exponents $e^+$ and $e^-$ are $l-1$ with $l \leq 1$ and their multiplicities are
\begin{equation}
p^{\epsilon}_{l-1,N}(x)=\sum^m_{u=0, u \, even}{^{\dagger}\widehat{h}^{(u)}_{(l-1)N}}\frac{x^u}{u!} \quad \text{ and } \quad q^\epsilon_{l-1,N}(x)=\sum^m_{u=0, u \, odd}{^{\dagger}\widetilde{h}^{(u)}_{(l-1)N}}\frac{x^u}{u!},
\end{equation}
with
\begin{align*}
&{^{\dagger}\widehat{h}^{(u)}_{(l-1)N}}=2  (\log q)^u({^{\dagger}h^{(u)}_{(l-1)N}}+\delta_{l,1}(c_u+ \, {^{\dagger} \lambda^{(u)}_1}  -\delta_{u,0}c_0)) \quad \text{ and } \\ \nonumber
&{^{\dagger}\widetilde{h}^{(u)}_{(l-1)N}}=-2 (\log q)^u \, {^{\dagger}h^{(u)}_{(l-1)N}}, \nonumber
\end{align*}
where $\dagger$ represents $c$ or $d$ depending on whether $\gm$ is $\ci$ or $\di$ and $P^{\epsilon}_{i,N}(0)=-2c_0$.

\item[(c)] Moreover, if $N$ is even, for $i=N/2$ the exponents $e^+$ and $e^-$ are $l-1/2$ with $l \geq 1$ and their multiplicities are
\begin{align}
&p_{l-1/2,N/2}(x)=\sum^m_{u=0, u \text{ even}} 2\frac{(x \log q)^u}{u!}{^{\dagger}h^{(u)}_{(l-1/2)N}} \quad \text{ and } \\
&q_{l-1/2,N/2}(x)=\sum^m_{u=0, u \text{ odd }} -2\frac{(x \log q)^u}{u!}{^{\dagger}h^{(u)}_{(l-1/2)N}}, \nonumber
\end{align}
where $\dagger$ represents $c$ or $d$ depending on whether $\gm$ is $\ci$ or $\di$.
\end{enumerate}
\end{proposition}


\begin{proof}
By Remark \ref{rmk:embeding}, part (a), we have that the embedding $\hat{\varphi}^{[m]}_s: \sqnh \longrightarrow \di$ is in fact the embedding given by proposition \ref{prop:fihat} composed by $T^{-1}$, where $T$ is the automorphism of $\gl$ defined in \eqref{eq:tpar}.

If $1 < i \leq [N/2]-\delta_{N,even}$, using \eqref{eq:deltafii} for the embedding in this case, we have
\begin{align*}
(\bigtriangleup_{m,1,\lambda,i})_n=&\lambda \bigg( \sum_{l \in \ZZ} \sum^m_{u=0} -\frac{(n \log q)^u}{u!}q^{n(1/2-l)}t^rE_{lN+1-i,lN+1-i} \\
&+\sum_{l \in \ZZ} \sum^m_{u=0}\frac{(-n \log q)^u}{u!}q^{n(-1/2+l)}t^rE_{(l-1)N+i,(l-1)N+i} \bigg) + \sum^m_{u=1} \eta_u(0,n)c_u.
\end{align*}
Making an adequate change of variable in $l$ and using \eqref{eq:di}, we have
\begin{align*}
(\bigtriangleup_{m,1,\lambda,i})_n=&  \sum_{l \in \ZZ} \sum^m_{u=0} \frac{(n \log q)^u}{u!} \big( - {^{d}\lambda^{(u)}_{lN+1-i}} q^{n(1/2-l)} \\
&+(-1)^r \, {^{d}\lambda^{(u)}_{(l-1)N+i}}q^{n(-1/2+l)} \big) + \sum^m_{u=1} \eta_u(0,n)c_u.
\end{align*}
Then,
\begin{align*}
&(\bigtriangleup_{m,1,\lambda,i})_n-(\bigtriangleup_{m,1,\lambda,i+1})_n= \\
&\sum_{l \geq 1} \sum^m_{u=0} \frac{(n \log q)^u}{u!}
({^dh^{(u)}_{lN-i}} \, q^{n(1/2-l)}+(-1)^u \,  q^{n(l-1/2)} \, {^dh^{(u)}_{(l-1)N+i}} ).
\end{align*}

Making use of the definitions of multiplicities and exponents for the quasipolynomial $P_i(x)$ in \eqref{eq:cuaspipolpn}, we finish the proof of (a).

Using \eqref{eq:deltafiN} for the embedding in this case, we obtain
\begin{align}\label{eq:5}
(\bigtriangleup_{m,1,\lambda,N})_n=& \lambda \bigg( \sum_{l \in \ZZ} \sum^{m}_{u=0}  \frac{(n \log q)^u}{u!} t^u \big( (-q^{n(1-l)}-(-1)^uq^{-n(1-l)})E_{(l-1)N+1,(l-1)N+1} \\ \nonumber
&+\sum_{l \in \ZZ} \sum^{m}_{u=0}  \frac{(n \log q)^u}{u!} t^u(q^{-nl}+(-1)^uq^{nl})E_{lN,lN} \big) \bigg)  \\ \nonumber
&+\sum^m_{u=1}2 \sinh_q(n/2)\eta_u(0,n)c_u. \nonumber
\end{align}
We make a change of variables in $l$. Using \eqref{eq:di} and the fact that
\begin{align*}
&2\eta_u(0,n)\sinh_q(n/2)=\frac{(1+(-1)^u)(n \log q)^u}{u!} \quad \text{ y } \\
&2\eta_0(0,n)\sinh_q(n/2)=\frac{2(n \log q)^u}{u!},
\end{align*}
 we have
\begin{align*}
(\bigtriangleup_{m,1,\lambda,N})_n=& \sum_{l \geq 1} \sum^{m}_{u=0}  \frac{(n \log q)^u}{u!} \big( (q^{-n(l-1)}+(-1)^uq^{n(l-1)}) \, {^dh_{(l-1)N}}+(1+(-1)^u) \, {^d\lambda^{(u)}_{1}} \big) \\
&+\sum^m_{u=0}\frac{(1+(-1)^u)(n \log q)^u}{u!}c_u-2c_0.
\end{align*}

In order to finish the proof, we study the parity of $u$ and split the sums accordingly. As a result of the definitions of multiplicities and exponents for the quasipolynomials $P^\epsilon_N(x)$ in \eqref{eq:cuaspipolpn}, we find the exponents and multiplicities expected for (b).

For $i=N/2$, following the same steps as in the proof of (a), we have
\begin{align*}
&(\bigtriangleup_{m,1/2,\lambda,N/2})_n-(\bigtriangleup_{m,1/2,\lambda,N/2+1})_n= \\
&\sum_{l \geq 1} \sum^m_{u=0} \lambda \bigg( -\frac{(n \log q)^u}{u!}t^u ( q^{n(1/2-l)}+(-1)^uq^{n(-1/2+l)})E_{(l-1/2)N+1,(l-1/2)N+1}\\
&+\frac{(n \log q)^u}{u!}t^u ( q^{n(1/2-l)}+(-1)^uq^{n(-1/2+l)})E_{(l-1/2)N,(l-1/2)N} \bigg).
\end{align*}
Making a change of variables in $l$ and using \eqref{eq:di}, we get
\begin{align*}
&(\bigtriangleup_{m,1/2,\lambda,N/2})_n-(\bigtriangleup_{m,1/2,\lambda,N/2+1})_n= \\
&\sum_{l \geq 1}\sum^m_{u=0} \frac{(n \log q)^u}{u!}( q^{n(1/2-l)}+(-1)^uq^{n(-1/2+l)}){ ^dh^{(u)}_{(l-1/2)N}}.
\end{align*}
Studying the parity of $u$ and splitting the sums accordingly, we find the exponents and multiplicities expected for this case, finishing the proof.

Consider now $\epsilon=-1.$ The embedding $\hat{\varphi}^{[m]}_s: \sqnh \longrightarrow \ci$ is in this case the embedding given by proposition \ref{prop:fihat} composed by $D=T \circ T^{\prime}$, where $T^{\prime}$ is the automorphism of $\gl$ defined in \eqref{eq:ta1}. Proceeding in an analogous way as for case $\epsilon=1$, we obtained the expected results.
\end{proof}


Consider an irreducible quasifinite highest weight $\sqnh$-module $V$ with central charge $c$ and generating series $\bigtriangleup_i(x)$ such that
\begin{align*}
P_{i}(n)=\bigtriangleup_{i,n}-\bigtriangleup_{i+1,n} \quad \text{ for } \quad  1 < i \leq [N/2]-\delta_{N, even} \quad  \text{ and } \\
P^{\epsilon}_N(n)=\bigtriangleup_{N,n} \quad \text{ for } \quad  n \neq 0 \quad \text{ and } \quad P^{\epsilon}_N(0)=-2c,
\end{align*}
where $P_i(x)$ are quasipolynomials and $P^{\epsilon}_N(x)$ are even quasipolynomials. Moreover, if $N$ is even, there exists an even quasipolynomial $P_{N/2}(x)$ such that
\begin{equation*}
P_{N/2}(n)=\bigtriangleup_{N/2,n}-\bigtriangleup_{N/2+1,n}.
\end{equation*}
Using the notation introduced in \eqref{eq:cuaspipolpn}, decompose the set
$A=\{s \in \CC| p_{s,i} \neq 0 \,\, \text{ for some } i \} \cup \{s \in \CC| p^{\epsilon}_{s,N} \neq 0 \} \cup \{s \in \CC| p_{s,N/2} \neq 0 \} \cup \{s \in \CC| q^{\epsilon}_{s,N} \neq 0 \} \cup \{s \in \CC| q_{s,N/2} \neq 0 \} $ into a disjoint union of equivalence classes under the condition
\begin{equation*}
s=q^a \sim q^{a^{\prime}}=s^{\prime} \Leftrightarrow a-a^{\prime} \in \ZZ+ \tau^{-1}\ZZ.
\end{equation*}
Pick a representative $s$ in an equivalence class $S$ such that $s=q$ if the equivalence class lies in $\ZZ$ and $s=q^{1/2}$ if the equivalence class lies in $\ZZ+1/2$. Let $S=\{q^a, q^{a+t_1}, q^{a+t_2}, \dots \} $ be such an equivalence class. Take $t_0=0$ and let \\
$m=\max_{s \in S} \{\deg \, p_{s,i}, \deg \, p^{\epsilon}_{s,N}, \deg \, q^{\epsilon}_{s,N}, \deg \, p_{s,N/2},\deg \, q_{s,N/2} \}$. It is easy to see that if $a=1$ or $a=1/2$, then $t_i \in \ZZ$.
Now, we will associate $S$ to a $\gm$-module $L^{[m]}_s(\lambda_S)$ in one of the following ways.

\begin{itemize}
\item
If $a \notin \ZZ/2$, for $1 < i \leq [N/2]-\delta_{N,even}$ let
\begin{align}
& {^ah^{(u)}_{(t_j-1)N+i}}=\frac{1}{(\log q)^u}\bigg( \frac{d}{dx} \bigg)^u p_{1/2-a+t_j,i}(0) \quad \text{ and} \\
& {^ah^{(u)}_{t_jN-i}}=\bigg(\frac{-1}{\log q}\bigg)^u \bigg( \frac{d}{dx} \bigg)^u p_{-1/2+a-t_j,i}(0),
\end{align}
and let
\begin{align}
& {^ah^{(u)}_{(t_j-1)N}}+\delta_{t_j,1}(c_u-\delta_{u,0}c_0)=\frac{1}{2(\log q)^u}\bigg( \frac{d}{dx} \bigg)^u p^{\epsilon}_{t_j,N}(0) \quad \text{ if $u$ even and} \\
& {^ah^{(u)}_{(t_j-1)N}}+\delta_{t_j,1}c_u=\frac{1}{2(\log q)^u}\bigg( \frac{d}{dx} \bigg)^u q^{\epsilon}_{t_j,N}(0) \quad \text{ if $u$ odd},
\end{align}
and if $N$ is even
\begin{align}
& {^ah^{(u)}_{(t_j-1/2)N}}=\frac{1}{(\log q)^u}\bigg( \frac{d}{dx} \bigg)^u p_{t_j,N/2}(0) \quad \text{ if $u$ even and} \\
& {^ah^{(u)}_{(t_j-1/2)N}}=-\frac{1}{(\log q)^u}\bigg( \frac{d}{dx} \bigg)^u q_{t_j,N/2}(0) \quad \text{ if $u$ odd},
\end{align}
for $u=0, \dots, m$. We associate $S$ to the $\glh$-module $L^{[m]}_s(\lambda_S)$ with central charges
\begin{equation*}
c_u=\sum_i \sum_{t_j}({^ah^{(u)}_{(t_j-1)N+i}}+{^ah^{(u)}_{t_jN-i}})+\sum_{t_j}({^ah^{(u)}_{(t_j-1)N}}+\delta_{N, \, even}{^ah^{(u)}_{(t_j-1/2)N}})
\end{equation*}
and labels
\begin{align*}
^a\lambda^{(u)}_l=&\sum_{(t_j-1)N+i \geq l}{^a\dot{h}^{(u)}_{(t_j-1)N+i}}+\sum_{t_jN-i \geq l}{^a\dot{h}^{(u)}_{t_jN-i}} \\
&+\sum_{(t_j-1)N \geq l}{^a\dot{h}^{(u)}_{(t_j-1)N}}+\delta_{N, \, even}\sum_{(t_j-1/2)N \geq l}{^a\dot{h}^{(u)}_{(t_j-1/2)N}},
\end{align*}
with ${^a\dot{h}^{(u)}_{t}}={^ah^{(u)}_{t}}-\delta_{t,0}c_u$.

\item
If $a=1/2$ and $N$ is even, for $1 < i \leq [N/2]-\delta_{N,even}$ let
\begin{align}
& {^dh^{(u)}_{(t_j-1/2)N+i}}=\bigg(\frac{-1}{\log q}\bigg)^u\bigg( \frac{d}{dx} \bigg)^u p_{t_j,i}(0) \quad \text{ if } \quad t_j > 0 \\
& {^dh^{(u)}_{(t_j-1/2)N-i}}= \frac{1}{(\log q)^u} \bigg( \frac{d}{dx} \bigg)^u p_{t_j,i}(0) \quad \text{ if } \quad t_j \leq 0,
\end{align}
and let
\begin{align}
& {^dh^{(u)}_{(t_j-1/2)N}}+\delta_{t_j,1}(c_u-\delta_{u,0}c_0)=\frac{1}{2(\log q)^u}\bigg( \frac{d}{dx} \bigg)^u p^{\epsilon}_{1/2-t_j,N}(0) \quad \text{ if $u$ even and } \\
& {^dh^{(u)}_{(t_j-1/2)N}}+\delta_{t_j,1}c_u=\frac{1}{2(\log q)^u}\bigg( \frac{d}{dx} \bigg)^u q^{\epsilon}_{1/2-t_j,N}(0) \quad \text{ if $u$ odd},
\end{align}
and if $N$ is even
\begin{align}
& {^dh^{(u)}_{t_jN}}=\frac{1}{2(\log q)^u}\bigg( \frac{d}{dx} \bigg)^u p_{t_j,N/2}(0) \quad \text{ if $u$ even and } \\
& {^dh^{(u)}_{t_jN}}=-\frac{1}{2(\log q)^u}\bigg( \frac{d}{dx} \bigg)^u q_{t_j,N/2}(0) \quad \text{ if $u$ odd},
\end{align}
for $u=0, \dots, m$. We associate $S$ to the $\di$-module $L^{[m]}_s(\lambda_S)$ with central charges
\begin{equation*}
c_u=\sum_i \sum_{t_j}({^dh^{(u)}_{(t_j-1/2)N+i}}+{^dh^{(u)}_{(t_j-1/2)N-i}})+\sum_{t_j}({^dh^{(u)}_{(t_j-1/2)N}}+\delta_{N, \, even}{^dh^{(u)}_{t_jN}})
\end{equation*}
and labels
\begin{align*}
^d\lambda^{(u)}_l=&\sum_{(t_j-1/2)N+i \geq l}{^dh^{(u)}_{(t_j-1/2)N+i}}+\sum_{(t_j-1/2)N-i \geq l}{^dh^{(u)}_{(t_j-1/2)N-i}} \\
&+\sum_{(t_j-1/2)N \geq l}{^dh^{(u)}_{(t_j-1/2)N}}+\delta_{N, \, even}\sum_{t_jN \geq l}{^dh^{(u)}_{t_jN}}.
\end{align*}

\item
If $a=1/2$, $N$ is odd for $1 < i \leq [N/2]-\delta_{N,even}$ let
\begin{align}
& {^bh^{(u)}_{(t_j-1/2)N+i-1/2}}=\bigg(\frac{-1}{\log q}\bigg)^u\bigg( \frac{d}{dx} \bigg)^u p_{t_j,i}(0) \quad \text{ if } \quad t_j > 0 \\
& {^bh^{(u)}_{(t_j-1/2)N-i-1/2}}= \frac{1}{(\log q)^u} \bigg( \frac{d}{dx} \bigg)^u p_{t_j,i}(0) \quad \text{ if } \quad t_j \leq 0,
\end{align}
and let
\begin{align}
& {^bh^{(u)}_{(t_j-1/2)N-1/2}}+\delta_{t_j,1}(c_u-\delta_{u,0}c_0)=\frac{1}{2(\log q)^u}\bigg( \frac{d}{dx} \bigg)^u p^{\epsilon}_{1/2-t_j,N}(0) \quad \text{ if $u$ even and} \\
& {^bh^{(u)}_{(t_j-1/2)N-1/2}}+\delta_{t_j,1}c_u=\frac{1}{2(\log q)^u}\bigg( \frac{d}{dx} \bigg)^u q^{\epsilon}_{1/2-t_j,N}(0) \quad \text{ if $u$ odd},
\end{align}
for $u=0, \dots, m$. We associate $S$ to the $\gm$-module $L^{[m]}_s(\lambda_S)$ with central charges
\begin{equation*}
c_u=\sum_i \sum_{t_j}({^bh^{(u)}_{(t_j-1/2)N+i-1/2}}+{^bh^{(u)}_{(t_j-1/2)N-i-1/2}})+\sum_{t_j}{^bh^{(u)}_{(t_j-1/2)N-1/2}}
\end{equation*}
and labels
\begin{align*}
^b\lambda^{(u)}_l=&\sum_{(t_j-1/2)N+i-1/2 \geq l}{^bh^{(u)}_{(t_j-1/2)N+i-1/2}}+\sum_{(t_j-1/2)N-i-1/2 \geq l}{^bh^{(u)}_{(t_j-1/2)N-i-1/2}} \\
&+\sum_{(t_j-1/2)N-1/2 \geq l}{^bh^{(u)}_{(t_j-1/2)N-1/2}},
\end{align*}
with $\gm=\tilde{b}^{[m]}_{\infty}$ if $\epsilon=1$ and $\gm=\bi$ if $\epsilon=-1$.
%

\item
If $a=1$, for $1 < i \leq [N/2]-\delta_{N, even}$, let
\begin{align}
& {^{\dagger}h^{(u)}_{t_jN-i}}=\frac{1}{(\log q)^u} \bigg( \frac{d}{dx} \bigg)^u p_{1/2-t_j,i}(0) \quad \text{ if } \quad t_j > 0 \\
& {^{\dagger}h^{(u)}_{(t_j-1)N+i}}=\bigg(\frac{-1}{\log q}\bigg)^u  \bigg( \frac{d}{dx} \bigg)^u p_{1/2-t_j,i}(0) \quad \text{ if } \quad t_j \leq 0,
\end{align}
and let
\begin{align}
& {^{\dagger}h^{(u)}_{(t_j-1)N}}+\delta_{t_j,1}(c_u+{^{\dagger}\lambda^{(u)}_1}-\delta_{u,0}c_0)=\frac{1}{2(\log q)^u}\bigg( \frac{d}{dx} \bigg)^u p^{\epsilon}_{t_j,N}(0) \quad \text{ if $u$ even and }  \\
& {^{\dagger}h^{(u)}_{(t_j-1)N}}=-\frac{1}{2(\log q)^u}\bigg( \frac{d}{dx} \bigg)^u q^{\epsilon}_{t_j,N}(0)  \quad \text{ if $u$ odd} ,
\end{align}
and if $N$ is even
\begin{align}
& {^{\dagger}h^{(u)}_{(t_j-1/2)N}}=\frac{1}{2(\log q)^u}\bigg( \frac{d}{dx} \bigg)^u p_{t_j-1/2,N/2}(0)  \quad \text{ if $u$ even and} \\
& {^{\dagger}h^{(u)}_{(t_j-1/2)N}}=-\frac{1}{2(\log q)^u}\bigg( \frac{d}{dx} \bigg)^u q_{t_j-1/2,N/2}(0)  \quad \text{ if $u$ odd} ,
\end{align}
for $u=0, \dots, m$, where ${\dagger}$ represents $d$ if $\epsilon=1$ and $c$ if $\epsilon=-1$. We associate $S$ to the $\gm$-module $L^{[m]}_s(\lambda_S)$ with central charges
\begin{equation*}
c_u=\sum_i \sum_{t_j}({^{\dagger}h^{(u)}_{t_jN+i}}+{^{\dagger}h^{(u)}_{t_jN-i}})+\sum_{t_j}({^{\dagger}h^{(u)}_{(t_j-1)N}}+\delta_{N, \, even}{^{\dagger}h^{(u)}_{(t_j-1/2)N}})
\end{equation*}
and labels
\begin{align*}
^{\dagger}\lambda^{(u)}_l=&\sum_{t_jN+i \geq l}{^{\dagger}h^{(u)}_{t_jN+i}}+\sum_{t_jN-i \geq l}{^{\dagger}h^{(u)}_{t_jN-i}} \\
&+\sum_{(t_j-1)N \geq l}{^{\dagger}h^{(u)}_{(t_j-1)N}}+\delta_{N, \, even}\sum_{(t_j-1/2)N \geq l}{^{\dagger}h^{(u)}_{(t_j-1/2)N}},
\end{align*}
where $\gm=\di$ and ${\dagger}=d$ if $\epsilon=1$ and $\gm=\ci$ with ${\dagger}=c$ if $\epsilon=-1$.

%
\end{itemize}

Denote $\{s_1,s_2, \dots \}$ with $s_i=q^{a_i}$ a set of representatives of equivalence classes of the set $A$. By Theorem \ref{teo:lirreduc}, the $\sqnh$-module $L^{[\vec{m}]}_{\vec{s}}(\lambda)$ is irreducible for $\vec{s}=(s_1,s_2, \dots)$ such that $a_i \enz$ implies that $a_i=1$ and $a_i \in \ZZ/2$ implies that $a_i=1/2$. Then, as consequence of the discussion above, the Theorem \ref{teo:lirreduc} and Propositions \ref{prop:casoanoenz/2}-\ref{prop:casoa1}, we have proved the following.

\begin{theorem}
Let $V$ an irreducible quasifinite highest weight $\sqnh$-module with central charge $c$ and let $P_i(x)$, $P^{\epsilon}_N(x)$ and, if $N$ is even, $P_{N/2}(x)$ the quasipolynomials given by Theorem \ref{eq:caracterizacion} written in the form \eqref{eq:cuaspipolpn}. Then, $V$ is isomorphic to the tensor product of the modules $L^{[m]}_s(\lambda_S)$ with distinct equivalence classes $S$.
\end{theorem}

\begin{remark}
A different choice of representative $s=q^a$ with $a \notin \ZZ/2$ in the equivalence class $S$ has the effect of shifting $\glh$ via the automorphism $\nu^i$ for some $i$. It is not difficult to see that any irreducible quasifinite highest weight module $L(\sqnh, \xi)$ can be obtained as above in an essentially unique way, up to this shift.
\end{remark}






\end{document}